\pdfoutput=1
\documentclass{article}
\usepackage[backend=biber,style=numeric,maxbibnames=9]{biblatex}
\usepackage[english]{babel}
\usepackage{palatino}
\bibliography{main}
\usepackage[letterpaper,top=2cm,bottom=2cm,left=3cm,right=3cm,marginparwidth=1.75cm]{geometry}

\usepackage{amsmath, amssymb, amsthm, algorithm, pseudo}
\usepackage{algorithmicx}
\usepackage[noend]{algpseudocode}
\usepackage{graphicx}
\usepackage{caption}
\usepackage{subcaption}
\usepackage{fullpage}
\usepackage[colorlinks=true, allcolors=blue]{hyperref}

\newtheorem{theorem}{Theorem}
\newtheorem{lemma}{Lemma}
\newtheorem{definition}{Definition}
\newtheorem{corollary}{Corollary}

\newtheorem{claim}{Claim}
\newcommand*\samethanks[1][\value{footnote}]{\footnotemark[#1]}
\newcommand{\cev}[1]

\allowdisplaybreaks

\title{Warm-starting Push-Relabel}

\author{Sami Davies\thanks{Department of EECS and the Simons Institute for the Theory of Computing, UC Berkeley. \texttt{samidavies@berkeley.edu}.}, Sergei Vassilvitskii\thanks{Google Research -- New York. \texttt{sergeiv@google.com}, \texttt{wangyy@google.com}.}, Yuyan Wang\samethanks}

\begin{document}
\maketitle

\begin{abstract}
Push-Relabel is one of the most celebrated network flow algorithms. 
Maintaining a pre-flow that saturates a cut, it enjoys better theoretical and empirical running time than other flow algorithms, such as Ford-Fulkerson.
In practice, Push-Relabel is even faster than what theoretical guarantees can promise, 
in part because of the use of good heuristics for seeding and updating the iterative algorithm. However, it remains unclear how to run Push-Relabel on an arbitrary initialization that is not necessarily a pre-flow or cut-saturating.
We provide the first theoretical guarantees for warm-starting Push-Relabel with a predicted flow,
where our learning-augmented version benefits from fast running time when the predicted flow is close to an optimal flow,
while maintaining robust worst-case guarantees.
Interestingly, our algorithm uses the \emph{gap relabeling heuristic}, which has long been employed in practice, even though prior to our work there was no rigorous theoretical justification for why it can lead to run-time improvements. We then provide experiments that show our warm-started Push-Relabel also works well in practice. 
\end{abstract}

\section{Introduction}
Maximum flow is a fundamental problem in combinatorial optimization. 
It admits many algorithms, from the famous Ford-Fulkerson algorithm~\cite{ford1956maximal} which employs augmenting paths, to recent near-linear time scaling based approaches~\cite{chen2022maximum}. In practice, however, the {\em push-relabel} family of algorithms is the benchmark for fast implementations~\cite{williamson2019, chandran2009}. 

Designed by Goldberg and Tarjan \cite{goldberg-tarjan}, the core Push-Relabel algorithm (Algorithm \ref{alg:vanilla-PR}) has running time $O(n^2m)$, where $n$ and $m$ are the number of vertices and edges in the network.
There are practical variants that reduce the running time to $O(n^2 \sqrt{m})$, and more theoretical adaptations that lead to sub-cubic $O(n m \log (n^2/m)$ run-time.  
Given the popularity of max-flow as a subroutine in many large scale applications~\cite{boykov2004experimental, kolmogorov2007applications, williamson2019}, it is no surprise that improving running times has been a subject of a lot of study, with multiple heuristic methods being developed~\cite{ahuja1997, cherkassky1995, goldberg2008partial}.

To complement the heuristics, researchers recently started looking at max-flow algorithms in the \emph{algorithms with predictions} framework~\cite{mitzenmacher2022algorithms}, 
and have successfully shown that one can improve running times when problem instances are not worst case, but share some commonalities ~\cite{polak2022learning, davies23b}.
Two independent groups initiated the study and proved that the running time of the Edmonds-Karp selection rule for Ford Fulkerson can be improved from $O(m^2n)$ to $O(m ||f^* - \hat{f}||_1)$, where $f^*$ is an optimal flow on the network and $\hat{f}$ is a predicted flow~\cite{davies23b, polak2022learning}.
These algorithms begin by modifying a predicted flow to form a feasible flow ~\cite{davies23b} or assuming that the predicted flow is already feasible ~\cite{polak2022learning}, and then start the augmenting path algorithms from that point. It is then relatively straightforward to bound the number of augmentations by the $\ell_1$-distance between the predicted and maximum flows.

While these works have been shown to improve upon the cold-start, non learning-augmented versions, it is important to note that they have been improving upon \emph{sub-optimal} algorithms for max flow. In this work, we show how to warm-start Push-Relabel, whose cold-start version is nearly state-of-the-art for the maximum flow problem.
This directly addresses the challenge specified in ~\cite{davies23b} on bringing a rigorous analysis for warm-starting non-augmenting path style algorithms.  In the process of doing so, we provide a theoretical explanation for the success of the popular \emph{gap relabeling heuristic} in improving the running time of Push-Relabel algorithms. 
Specifically, both the gap relabeling heuristic and our algorithm maintain a  cut with \emph{monotonically decreasing t-side nodes} (see Section \ref{sec:tech} for more), which directly leads to improved running times for our version of Push-Relabel and the gap relabeling heuristic. 
Lastly, we show that our theory is predictive of what happens in practice with experiments on the image segmentation problem.

\subsection{Preliminaries}
\textbf{Graph, flow and cut concepts.} Our input is a network $G=(V,E)$,
where each directed edge $e \in E$ is equipped with an integral capacity $c_e \in \mathbb{Z}_{\geq 0}$. 
Let $|V|=n$ and $|E|=m$.
$G$ contains nodes $s$, the source, and $t$, the sink.  $G$ is connected: $\forall u \in V$, there are both $s-u$ and $u-t$ paths in $G$.
A flow $f\in \mathbb{Z}_{\geq 0}^m$ is feasible if it satisfies: (1) \emph{flow conservation}, 
meaning any $u \in V \setminus \{s, t \}$ satisfy 
$\sum_{(v,u) \in E}f_e = \sum_{(u,w) \in E}f_e$; 
(2) \emph{capacity constraints}, meaning for all $e \in E$, $f_e \leq c_e$. Our goal is to find the maximum flow, i.e. one with the largest amount of flow leaving $s$. 

We call $f$ a \emph{pseudo-flow} if it satisfies capacity constraints only. A node $u \in V \setminus \{s,t\}$ is said to have \emph{excess} if it has more incoming flow than outgoing, i.e., $\sum_{(v,u) \in E}f_e > \sum_{(u,w) \in E}f_e$; analogously it has \emph{deficit} if its outgoing flow is more than ingoing. 
We denote the excess and deficit of a node $u$ with respect to $f$ as 
$\textsf{exc}_f(u) = \max\{\sum_{(v,u) \in E}f_e - \sum_{(u,w) \in E}f_e,0\}$ and $\textsf{def}_f(u)=\max\{\sum_{(u,w) \in E}f_e-\sum_{(v,u) \in E}f_e ,0\}$, where at most one can be positive. A pseudo-flow can have both excesses and deficits, whereas a \emph{pre-flow} is a pseudo-flow with excess only.

For a pseudo-flow $f$,
the \emph{residual graph} $G_f$ is a network on $V$;
for every $e=(u,v) \in E$, $G_f$ has edge $e$ with capacity $c'_e=c_e-f_e$ 
and a backwards edge $(v,u)$ with capacity $f_e$.
Let $E(G_f)$ denote the edges in $G_f$.
The value of a pseudo-flow $f$ is $\textsf{val}(f) = \sum_{e= (s,u)}f_e$, the total flow going out of $s$. Notice that this is not necessarily equivalent to the total flow into $t$ since flow conservation is not satisfied.
A \emph{cut-saturating} pseudo-flow is one that saturates some $s-t$ cut in the network. Push-Relabel maintains a cut-saturating pre-flow; equivalently, there is no $s-t$ path in the residual graph of the pre-flow. We use $\delta(S, T)$ to denote an $s-t$ cut between two sets $S$ and $T$.
Note that the cut induced by any cut-saturating pseudo-flow $f$ can be found
by taking $T = \{u\in V: \exists  u-t \text{ path in }G_f\}$ (including $t$) and $S = V \setminus T$.

\textbf{Prediction.}
The prediction that we will use to seed Push-Relabel is some $\hat{f}\in \mathbb{Z}^m_{\geq 0}$, which is a set of integral values on each edge. 
Observe that one can always cap the prediction by the capacity on every edge to maintain capacity constraints, so throughout this paper we will assume $\hat{f}$ is a pseudo-flow.
It is important to note that our predicted flow is \emph{not} necessarily feasible or cut-saturating, and part of the technical challenge is making use of a good predicted flow despite its infeasibility. 

\textbf{Error metric.} We measure the error of a predicted pseudo-flow $\hat{f}$ on $G$. The smaller the error is, the higher quality the prediction is, and the less time Push-Relabel seeded with $\hat{f}$ should take. A pseudo-flow becomes a maximum flow when it is both feasible and cut-saturating. Hence, the error measures how far $\hat{f}$ is from being cut-saturating while being feasible.
We say that a pseudo-flow $\hat{f}$ is $\sigma$ \emph{far from being cut-saturating} if there exists a feasible flow $f'$ on $G_{\hat{f}}$ where \textsf{val}$(f') \leq \sigma$ and $\hat{f}+f'$ is cut-saturating on $G$ (though the cut does not have to be a min-cut). To measure how far $\hat{f}$ is from being feasible, we sum up the total excesses and deficits. 
In total we use the following error metric:

\begin{definition}\label{def:error}
For pseudo-flow $\widehat{f}$ on network $G$,
the \emph{error} of $\widehat{f}$ is the smallest integer $\eta$
such that (1) $\widehat{f}$ is $\eta$ far from being cut-saturating
and (2) the total excess and deficit in $G$ with respect to $\widehat{f}$  is $\sum_{u \in V \setminus \{s,t\}} \textsf{exc}_{\widehat{f}}(u)+\textsf{def}_{\widehat{f}}(u) \leq \eta$.
\end{definition}

If $\eta=0$, $\hat{f}$ is the max-flow and the cut that is saturated is the min-cut.
The previously studied error metric for predicted flows, such as by   \cite{davies23b} and  \cite{ polak2022learning}, was $||f^*-\hat{f}||_1$, for any max-flow $f^*$.

PAC-learnability is the standard to justify that the choice of prediction and error metric are reasonable. Flows are
PAC-learnable with respect to the $\ell_1$-norm \cite{davies23b}.
Our results hold replacing our error metric with the $\ell_1$-norm because 
our metric provides a more fine-grained guarantee than the $\ell_1$-norm (i.e., if a prediction $\hat{f}$ has error $\eta$, then  $\eta \leq ||f^*-\hat{f}||_1$). Thus we can omit any  theoretical discussion of learnability.
We present this work with respect to our error metric as we find the $\ell_1$ error metric to be unintuitive, in the sense that it is not really descriptive of how good a predicted flow is.

\paragraph{Push-Relabel.} Here, we review the ``vanilla" form of Push-Relabel. The Push-Relabel algorithm maintains a pre-flow and set of valid \emph{heights} (also called labels).
Heights $h : V \rightarrow \mathbb{Z}_{\geq 0}$
are \emph{valid} for a pre-flow $f$ if for every edge in the residual network $(u,v) \in E(G_f)$, $h(u) \leq h(v)+1$, and if $h(s)=n$ and $h(t)=0$. 
An edge $(u,v) \in E(G_f)$  is called \emph{admissible} if $h(u) = h(v)+1$ and $c'_{(u,v)} > 0$, which means we can push flow from $u$ to $v$. 
The formal Push-Relabel algorithm is in Algorithm \ref{alg:vanilla-PR}, seeded with
$f^{\textsf{init}}$ where $f^{\textsf{init}}_e=c_e$ for all $e = (s,u)$ and otherwise $f^{\textsf{init}}_e=0$.

It is known from the original analysis that all heights in Push-Relabel are bounded by $2n$. 
\begin{lemma}\label{lem: 2n-heights}
For a pre-flow $f$ on network $G$,
every node $u$ with $\textsf{exc}_f(u) >0$ has a path in $G_f$ to $s$. 
Further, for $d(u,v)$ the length of the shortest path between $u$ to $v$ in $G_f$, any valid heights in Push-Relabel (Algorithm \ref{alg:vanilla-PR}) satisfy
$$h(u) \leq h(v) + d(u,v)$$
Choosing $v=s$, we have $h(u) \leq h(s) + n = 2n$.
\end{lemma}

\iffalse
We show that when the given $G$ is a connected graph everywhere, given any pseudo-flow $f$ it must have either a path to $s$, or to $t$, or both in the residual graph.
\begin{lemma}
\label{lem:connection}
Given any pseudo-flow $f$, all nodes with no path to $s$ or to $t$ form a subgraph that is originally disconnected from $s$ in $G$.
\end{lemma}
\begin{proof}
Say there is a inclusion-wise maximal set $S$ of such nodes. We prove that there is no edge $(p,q)$ in the residual graph of $f$ such that $p \in S$ and $q \notin S$. Say there exists such an edge. If $q$ has a path to $s$ or to $t$, then $p$ also have a path to $\{s,t\}$, contradicting definition of $S$; on the other hand if $q$ doesn't have a path to $s$ or $t$, $q$ should be in $S$, contradicting the maximality of $S$.

Hence there cannot be any flow into $S$, otherwise some reverse edge $(p,q)$ from $S$ to $V \setminus S$ would be in the residual graph. We then show in graph $G$ there is no edge from $V \setminus S$ to $S$ at all. If such an edge exists, it would be in the residual graph since there is no in-flow. 

Therefore, $S$ is a set of nodes with no path to $t$ in $G$ right from the start.
\end{proof}
\fi

\iftrue
\begin{algorithm}[t!] 
\caption{Push-Relabel}\label{alg:vanilla-PR}
\begin{algorithmic}
\State {\bfseries Input}: Network $G$
\State  Define $f_{e}=c_{e}$ for $e = (s,u)$ and $f_e=0$ for all other $e$
\State  Define $h(u)=0$ for all $u \in V \setminus \{s\}$ and $h(s)=n$
\State  Build residual network $G_f$
\While{ $\exists$ node $u$ with $\textsf{exc}_f(u)>0$}
\If{ $\exists$ admissible $(u,v) \in E(G_f)$ with $f_{(u,v)} < c_{(u,v)}$}
% \STATE Send flow $\min\{c_{(u,v)}-f_{(u,v)},\textsf{exc}_f(u)\}$ on $(u,v)$
\State   Update $f$ by sending an additional flow value of $\min \{\textsf{exc}_f(u), c'_{(u,v)}\}$ along $(u,v)$
\State  Update $G_f$
\Else
\State  Update $h(u) = 1 + \min_{v : (u,v) \in E(G_f)}h(v)$
\EndIf
\EndWhile
\State  {\bfseries Output}: $f$
\end{algorithmic}
\end{algorithm}
\fi
At any point of the algorithm, 
the $s-t$ cut maintained by the pre-flow can be induced using the heights.
\begin{lemma}\label{lem:induced-cut}
    For $h$ valid heights for a cut-saturating pseudo-flow $f$ on network $G$,
    let $\theta$ be the smallest positive integer such that $\theta \notin \{h(u)\}_{u \in V}$. 
    Then $S=\{u \in V: h(u) > \theta\}$ and $T=\{u \in V: h(u) < \theta\}$ form a cut saturated by $f$.
\end{lemma}

We will call this the cut \emph{induced by the heights}.
Indeed, such a threshold $\theta$ can be found because $\{h(u)\}_{u \in V}$ has at most $n$ different values, but $h(s) = n$ and $h(t) = 0$, so among the $n+1$ values $\{0, 1, \ldots, n\}$, there is at least one not in the set. It is easy to see $\delta(S,T)$ is a saturated cut. For any $u \in S, v \in T$, we have that $h(u) > h(v) + 1$, so $(u,v)$ is not admissible in the residual graph. It follows that either $(u,v) \in E$ and it is saturated, or $(v,u) \in E$ and $f_{(v,u)} = 0$.

A saturated cut be can defined from the set of vertices that can reach the sink in the residual graph. 
\begin{lemma}
\label{lem:S_T_cut}
For any pseudo-flow $f$ on network $G$, let $T$ be all nodes that can reach $t$ in $G_f$ and $S = V \setminus T$. If both $S,T$ are non-empty, then $\delta(S,T)$ is a saturated cut.
\end{lemma}
\begin{proof}
Fix $u \in S$, $v \in T$. Since $v$ can reach $t$ and $u$ cannot,
any edge $(u,v)$ from $S$ to $T$ in $G$ must be saturated by $f$,
and any edge $(v, u)$ from $T$ to $S$ in $G$ must have no flow. 
This is because if either of these were not true, 
the edge $(u,v)$ in $G_f$ would have positive capacity, allowing $u$ to reach $t$.
Hence $\delta(S,T)$ is saturated by $f$.
% by definition of $T$ edge $(v, u)$ in the residual graph. Otherwise, since $v$ can reach $t$, $u$ can also reach $T$ and $u$ should be in $T$ instead of $S$. It follows that  
\end{proof}

Lemmas \ref{lem:induced-cut} and \ref{lem:S_T_cut} apply to all pseudo-flows, whereas vanilla Push-Relabel must take a pre-flow as input. Before this work, it was unclear how to seed Push-Relabel with anything other than $f^{\textsf{init}}$.

Push-Relabel can be implemented with the \emph{gap relabeling heuristic},
i.e., whenever there is some integer $0<\theta<n$ with no nodes at height $\theta$, 
then nodes with height between $\theta$ and $n$ have their height increased to $n$. 
See Algorithm \ref{alg:ws_pr} for the formal statement,
where the cold-start version of this algorithm is to take as
input the cut-saturating pre-flow $f^{\textsf{init}}$.

\subsection{Technical contribution}
\label{sec:tech}
We first review Push-Relabel with the gap relabeling heuristic when the algorithm is seeded with a prediction that is a cut-saturating pre-flow with error $\eta$. We show that this version of Push-Relabel finds an optimal solution in time $O(\eta \cdot n^2)$.
Recall that in this setting, Definition \ref{def:error} implies that $\eta$ is just the total excess. This running time also holds for cold-start versions of the algorithm when the max-flow/min-cut value is known to be bounded by $\eta$. 
This is (1) the first theoretical analysis of the gap relabeling heuristic, and, (2) the first result showing a running time bounded by the volume of the cut in Push-Relabel.  Unlike the Ford-Fulkerson algorithm, which admits a naive run-time bound of $O(\eta \cdot m)$ when the max-flow value is bounded by $\eta$, an analogous claim cannot be made easily for Push-Relabel.

Intuitively, Push-Relabel with the gap relabeling heuristic 
essentially maintains a cut whose $t$-side
is \emph{monotonically decreasing} (i.e., it moves nodes on the $t$-side of the cut to the $s$-side, but not the other way around), and resolves excess on the $t$-side by routing excess
flow to $t$, or updating the cut so the excess node is on the $s$-side of the new cut. 
In the latter case, the excess flow will be sent back to $s$ later. 
The same insight will be used in our general warm-started version of Push-Relabel that can be seeded with any pseudo-flow.

Our main result is the following theorem, 
which applies in the general setting where the prediction is any pseudo-flow, i.e., the prediction is not necessarily a pre-flow and is not necessarily cut-saturating.

\begin{theorem}\label{thm:main}
Given a predicted pseudo-flow $\widehat{f}$ with error $\eta$ on network $G$, 
there exists a warm-start version of Push-Relabel that obtains the minimum cut in time $O(\eta \cdot n^2)$.
\end{theorem}

Our warm-start version of Push-Relabel has several phases. 
Within each phase, Push-Relabel with the gap relabeling heuristic is used as a subroutine on auxiliary graphs.

First, we show that one can modify the prediction $\widehat{f}$ to be a cut-saturating pseudo-flow; we call this cut-saturating pseudo-flow $f$. This is accomplished by running the cold-start Push-Relabel with the gap relabeling heuristic on the residual graph to find a max-flow/min-cut.

\begin{figure}
    \centering
    \includegraphics[width=\linewidth]{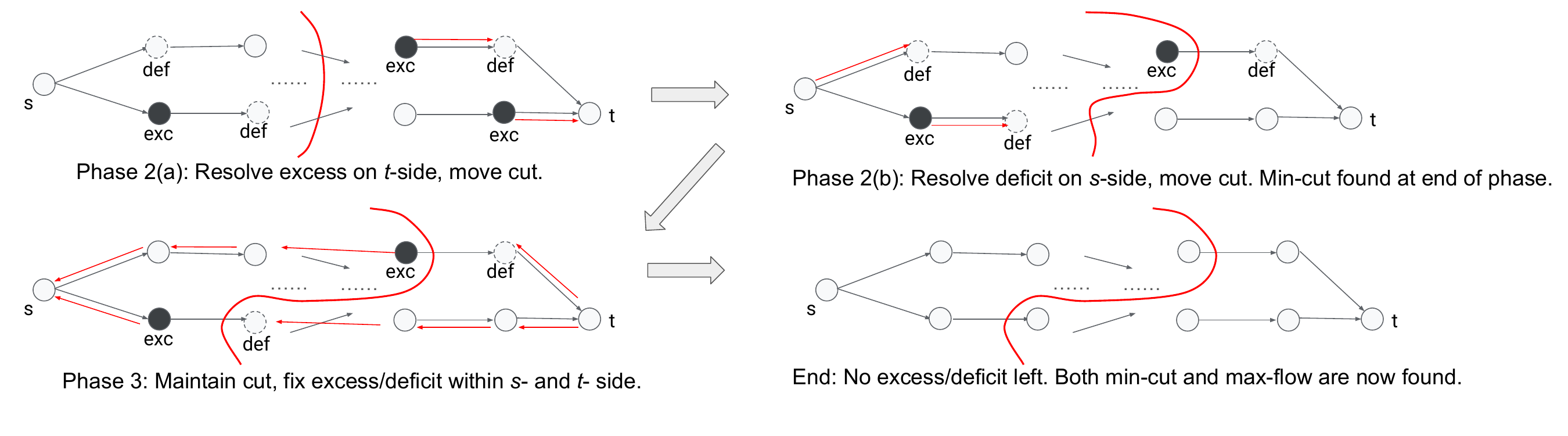}
    \caption{An illustration of different phases of warm-start Push-Relabel, starting with a cut-saturating pseudo-flow. The red curve denotes the cut. The black arrows denote the existing flows, whereas the red arrow denotes the flows sent in each phase to resolve excesses/deficits. Notice that as flows are sent, new edges become saturated and smaller cuts are found, swapping excess and deficit nodes to the opposite sides of the cut.}
    \label{fig:flow}
\end{figure}

We begin the second phase by
routing flow 
\emph{within} the two sides of the cut induced by the pseudo-flow
to resolve some of the excess and deficit. 
The maintained cut gradually changes as we send flow from node to node, and push certain nodes to different sides of the cut. 
We do this by running the standard, cold-start version of Push-Relabel on auxiliary networks on each side of the cut, 
and then adding the resulting flow to $f$. 
Because these auxiliary networks have small minimum cut value, 
the flows can be found quickly.
We continue changing the cut until all excess nodes end up on the $s$-side of the cut and all deficit nodes end up on the $t$-side of the cut.  
This ``swapping" procedure between excess and deficits nodes between the $s$- and $t$- sides of the cut is our biggest technical innovation.
Either the excesses are resolved within the $t$-side of the cut,
or we find a new cut between the $t$-side excess nodes and the $t$-side deficit nodes plus $t$.  
We modify the cut in $G$ accordingly to separate all excess from the $t$-side, 
which also results in a cut whose $t$-side is monotonically decreasing---an interesting point which we show also occurs in the cut maintained by the gap relabeling heuristic.
A mirrored version of this process is performed on the $s$-side of the cut.

In the final phase, we have a new cut-saturating pseudo-flow with all excess nodes
on the $s$-side of the cut and all deficit nodes on the $t$-side of the cut. 
This cut is actually a min-cut. On the $s$-side, the excess nodes send flow to the source, and on the $t$-side, the sink sends flow to deficit nodes (hence removing existing flows).
The result is a max-flow.
See Figure \ref{fig:flow} for an illustration of phases 2 and 3.

In Section \ref{sec: empirical}, we run our warm-start Push-Relabel compared to a cold-start version. We see that the warm-start improves over the cold-start by a larger percentage as the size of the image increases.

\subsection{Related work}
\label{sec:related_work}

Push-Relabel is one of the most popular algorithms for finding a max flow or min cut.
See Algorithm \ref{alg:vanilla-PR} for its statement.
The algorithm has running time $O(n^2m)$ and was designed by Goldberg and Tarjan \cite{goldberg-tarjan}.

Several computational studies focus on the performance of Push-Relabel.
A very popular heuristic is to choose the push operation to occur from the node with the largest height, 
and this gives theoretical improvements too, with running time $O(n^2 \sqrt{m})$.
In a well-known empirical study performed in 1997, largest-height Push-Relabel was the fastest max-flow algorithm on most classes of networks, outperforming Dinic's augmenting flow algorithm, Karzanov's algorithm, and Ahuja, Orlin, and Tarjan's excess--scaling algorithm \cite{ahuja1997}.
Since then, Hochbaum's Pseudo-flow algorithm---and extensions of it---have been shown in experiments to be faster than largest-height Push-Relabel \cite{hochbaum2008, chandran2009, goldberg2015faster}.
Throughout the development of new max-flow algorithms, Push-Relabel remains one of the most versatile and is often the standard benchmark to which new flow algorithms are compared.

Two notable heuristics for Push-Relabel are the global relabeling heuristic and gap relabeling heuristic \cite{johnson1993network, cherkassky1995, goldberg2008partial}. 
The \emph{global relabeling heuristic} occasionally updates the heights to be 
a node's distance from $t$ in the current residual graph.
Interestingly, these heuristics are more effective together than separately in practice \cite{cherkassky1995}.

The field of learning-augmented algorithms, 
also known as algorithms with predictions, 
has gained notable popularity over the past 5 years. 
An algorithm is given access to a prediction about a quantity pertaining to the input, 
and this prediction can guide the algorithm into making better choices. 
Predictions have primarily been used to improve competitive ratios
for online algorithms for problems in many areas, including scheduling \cite{lassota2023minimalistic, ahmadian2023robust}, caching \cite{lykouris2021competitive,bansal2022learning}, and data structures \cite{kraska2018case, lin2022learning}.
More recently, they have also been used to improve the running times of algorithms \cite{dinitz2021faster, chen2022faster, oki2023faster}.
Two independent groups initiated the study of learning-augmented max-flow by warm-starting Ford-Fulkerson procedures \cite{polak2022learning, davies23b}. 
These works show that the running time of Edmonds-Karp can be improved from $O(m^2n)$ to $O(m ||f^* - \widehat{f}||_1)$, for $f^*$ an optimal flow on the network and $\widehat{f}$ the prediction; 
experiments on image segmentation instances exemplify that the theory is predictive of practice \cite{davies23b}.

More on the many max-flow algorithms,can be found in the survey by Williamson \cite{williamson2019}, 
while more on learning-augmented algorithms can be found in the survey by
Mitzenmacher and Vassilvitskii \cite{mitzenmacher2022algorithms}. 

\section{Gap Relabeling Push-Relabel: Cold- and Warm-Start}
\label{sec:ws_pr}
Among the many heuristic adaptations for Push-Relabel,
the gap relabeling heuristic is known to empirically improve the performance. In this section, we analyze the performance of
Push-Relabel with the gap relabeling heuristic
(Algorithm \ref{alg:ws_pr}) when given a cut-saturating pre-flow $f$, and tie the running time to the error of $f$.

\begin{algorithm}[t]
\begin{algorithmic}
    \State  \textbf{Input:} Network $G$, a cut-saturating pre-flow $f$
    \State Construct residual network $G_f$ with capacity $c'$
    \State  Run Algorithm \ref{alg:initialheights} on $G$ and $f$, obtain $h$ 
    and $(S,T)$
    \State  Initialize $\theta=\min\{z \in \mathbb{Z}_{>0} :\nexists u \in T \text{ with } z=h(u) \}$
    \While {$\exists u \in T$ with $\textsf{exc}_f(u)>0$} 
    \If {$\exists v \in T$ with $h(u) = h(v) + 1$, $(u, v) \in E(G_f)$, and $f_{(u,v)} < c_{(u,v)}$}
    \State   Update $f$ by sending an additional flow value of $\min \{\textsf{exc}_f(u), c'_{(u,v)}\}$ along $(u,v)$
    \State  Update $G_f$
    \Else 
    \State  Relabel $u$ with $h(u)= \min_{v: (u, v) \in E(G_f) } h(v) + 1$
    \State  Update $\theta = \min\{z \in \mathbb{Z}_{>0} :\nexists u \in T \text{ with } z=h(u) \}$
    \For{$p \in T$ with $h(p) > \theta$}
        \State  Remove $p$ from $T$, add $p$ to $S$
        \State  Update $p$'s height to $h(p)=n$
    \EndFor
    \EndIf
    \EndWhile
    \State  Take $G_f$ as input and run Algorithm \ref{alg:vanilla-PR} on it to fix excesses, outputs flow $f^*$
    \State  Return cut parts $S$ and $T$, cut $\delta(S,T)$, and flow $f+f^*$
\end{algorithmic}
\caption{Warm-start Push-Relabel with Gap Relabeling}
\label{alg:ws_pr}
\end{algorithm}

Algorithm \ref{alg:ws_pr} begins by running Algorithm \ref{alg:initialheights} as a subroutine to find the $s-t$ cut saturated by $f$ and define valid heights for $f$ which also induce that cut. Algorithm \ref{alg:initialheights} runs a BFS in the residual graph to find all nodes that have a path to $t$ and names this set $T$. The other nodes belong to $S$. The cut $\delta(S,T)$ has to be saturated by $f$ (see Lemma \ref{lem:S_T_cut}).

From there on, Algorithm \ref{alg:ws_pr} has a two-phase structure. 
In phase one (the main \textsf{WHILE} loop), 
the algorithm maintains a set of $t$-side nodes of the cut, denoted by $T$, and all heights in $T$ must cover a series of consecutive numbers starting from $0$. Intuitively, for any node with excess in $T$, Algorithm \ref{alg:ws_pr} tries to resolve its excess by re-routing it to other nodes. 
If this is not possible, the algorithm forces the node (and possibly other nodes, too) to leave $T$ and changes the cut maintained by the pre-flow.
The cut only changes when a node is relabeled in a way that results in a break in the series of consecutive heights starting from $0$ in $T$, where the smallest missing height is denoted by $\theta$. The algorithm then removes all nodes from $T$ with height bigger than $\theta$; importantly, these nodes will \emph{never} enter $T$ again. Phase one terminates when $T$ has no more excess nodes. With the correct data structure, the threshold height $\theta$ and the set $T$ can be maintained at minimal cost.

At the end of phase one, 
despite potential excess nodes on the $s$-side, 
the cut obtained is a min-cut.
\begin{lemma}
\label{lem:pre_flow_no_excess_T}
Let $f$ be a pre-flow saturating cut $\delta(S, T)$ on network $G$. 
If there are no excess nodes in $T$, 
then all excess in $S$ can be sent back to $s$ without crossing the cut, implying that the cut $\delta(S, T)$ is a min-cut.
\end{lemma}
\begin{proof}
It is known from the proof of the vanilla Push-Relabel algorithm that all excess nodes in a pre-flow must have a path back to $s$; see Lemma \ref{lem: 2n-heights}. When $f$ saturates $\delta(S,T)$, such a path cannot go from $S$ to $T$, so the path must be within $S$. The last two lines of Algorithm \ref{alg:ws_pr} will resolve the excesses without effecting the saturated cut. So we have a feasible flow saturating a cut, meaning the flow is a max-flow and the cut is a min-cut.
\end{proof}
For applications where max-flow is simply a subroutine for finding a min-cut rather than the goal---such as in \emph{image segmentation}---
in Algorithm \ref{alg:ws_pr}, one can omit running Algorithm \ref{alg:vanilla-PR} after the \textsf{WHILE} loop ends.
\iftrue
\begin{algorithm}[t!] \label{alg:height}
\caption{Define Heights}\label{alg:initialheights}
\begin{algorithmic}
\State  {\bfseries Input}: Network $G$, a cut-saturating pseudo-flow $f$
    \State  Define $h(s) = n$, $h(t)=0$
    \State  Run BFS in $G_{f}$
    \State  Let $T = \{u \in V : \exists u-t \text{ path in }G_f\}$
    \State  Let $S = V \setminus T$
    \For{ all $u \in S$}
    \State  Let $h(u)=n$
    \EndFor
    \For{ all $u \in T$}
    \State  Let $h(u)$ be the shortest path length from $u$ to $t$ in $G_f$
    \EndFor
    \State {\bfseries Output}: Valid heights $h$, and cut parts $S$ and $T$
\end{algorithmic}
\end{algorithm}
\fi

We show that the running time is tied to $\eta$, which, in this case, is the total excess in $f$.
\begin{theorem} \label{thm:pre_flow_ws}
Given a cut-saturating pre-flow $f$ with error $\eta$ on network $G$,
Algorithm \ref{alg:ws_pr} finds a max-flow/min-cut in  running time $O(\eta \cdot n^2)$.
\end{theorem}

\begin{proof}
The algorithm first works to resolve excess in $T$, possibly moving nodes from $T$ to $S$ to do so. 
Once all excess is in $S$, correctness follows from Lemma \ref{lem:pre_flow_no_excess_T}.
Note that the conditions of  Lemma \ref{lem:pre_flow_no_excess_T} are satisfied since
by Lemma \ref{lem:S_T_cut} the cut output by Algorithm \ref{alg:initialheights} is saturated by $f$.

To bound the running time of Algorithm \ref{alg:ws_pr},
we use a potential function argument 
that is different from than in the standard Push-Relabel analysis.

We first bound the running time of the main \textsf{WHILE} loop that terminates when all excess is contained in $S$ and the min-cut is found.
We define the potential function 
$\Phi(T) = \sum_{u \in T}\textsf{exc}_f(u) \cdot h(u)$.
The operations involved change the value of $\Phi(T)$ in the following way.
\begin{itemize}
\item Saturated/Unsaturated push: 
In either case, at least one unit of excess flow is pushed from a higher height to a lower height, since for edge $(u,v)$ to be admissible, $h(u) = h(v) + 1$. 
Therefore, $\Phi(T)$ decreases by at least 1. 
 
\item Relabeling: Any relabeling operation increases $\Phi(T)$. 
However, the total of all of these increases is at most $\eta \cdot n^2$.
The $\eta$ term upper bounds the possible excess at any node,
whereas the $n^2$ term is because a node's (of which there are at most $n$) height only ever increases, 
and the height cannot increase beyond $n$ before it must leave $T$ permanently. 

\item Removing nodes from $T$: Decreases $\Phi(T)$. 
\end{itemize}
Hence the total running time before finding the min-cut is bounded by $O(\eta \cdot n^2)$.

To bound the time for finding the max-flow, notice that the total excess in $G$ only decreases, so when we start to route excesses in $S$ to $s$, the total excess is also bounded by $\eta$.  
The same potential function argument can be used to prove it also takes $O(\eta \cdot n^2)$ time to resolve all excess in $S$, 
though using the fact that in Push-Relabel,
heights are always bounded by $2n$ (see Lemma \ref{lem: 2n-heights}).
\end{proof}
Although Algorithm \ref{alg:ws_pr} is presented as being seeded with an existing pre-flow, the same bound applies to the cold-start gap relabeling Push-Relabel when the min-cut of $G$ is at most $\eta$. 
This will prove useful in Section \ref{sec:ws_pseudoflow},
as we repeatedly use Algorithm \ref{alg:ws_pr} as a subroutine to fix excess and deficits and redefine cuts on networks with small cut value.
\begin{corollary}[cold-start run-time with $\eta$ min-cut value]
\label{cor:cs_pr}
If network $G$ is known to have a max-flow/min-cut value of at most $\eta$, 
one can use Algorithm \ref{alg:ws_pr} to 
obtain a max-flow and min-cut for $G$ in running time $O(\eta \cdot n^2)$.
\end{corollary}
\begin{proof}
Create an auxiliary graph $G'$ by taking a copy of $G$ and adding a super-source $s^*$ and an edge $(s^*, s)$ with capacity $\eta$. Create a pre-flow $f^{\textsf{init}}$ on $G'$ by saturating $(s^*, s)$ and letting $f_e=0$ on all other edges in $G'$.
Now run Algorithm \ref{alg:ws_pr} with inputs $G'$ and $f^{\textsf{init}}$.
The initial (and maximum) excess in $G'$ was $\eta$, 
and so the run-time is bounded by $O(\eta \cdot n^2)$, 
as in the proof of Theorem \ref{thm:pre_flow_ws}.
\end{proof}

Note that we only assume $\eta$ to be known for simplicity of argument. A slightly modified algorithm can achieve the same running time with unknown $\eta$. See the discussion in Appendix \ref{sec:predict-cut}.

\section{Warm-starting Push-Relabel with General Pseudo-flows}
\label{sec:ws_pseudoflow}
We extend the results in Section \ref{sec:ws_pr} to when the given prediction is a general pseudo-flow $\widehat{f}$ as opposed to a cut-saturating pre-flow, i.e.,  
$\widehat{f}$ may not be cut-saturating and may have deficit nodes. 
We assume $\widehat{f}$ has $\eta$ error, as defined in Definition \ref{def:error}. 
The first phase of our algorithm augments $\widehat{f}$ by finding an $s-t$ flow to add to 
$\widehat{f}$ so that the resulting pseudo-flow saturates a cut. 
Then, in phase two, it sends flow within both sides of the cut to eliminate and swap excess/deficit nodes, until all excess nodes are on the $s$-side of the cut and all deficit nodes are on the $t$-side. The min-cut is found at this point. 
Finally in phase three, the algorithm sends the remaining excess to $s$ and deficit to $t$ to obtain a feasible flow, which is also a max-flow.

\subsection{Obtaining a cut-saturating pseudo-flow from $\hat{f}$}

The first phase is to pre-process $\widehat{f}$ into a cut-saturating pseudo-flow on $G$. See Algorithm \ref{alg:pre-process}.

\iftrue
\begin{algorithm}[t!]
\caption{Find a Cut-saturating Pseudo-flow}\label{alg:pre-process}
\begin{algorithmic}
\State  {\bfseries Input}: Network $G$, a pseudo-flow $\widehat{f}$
    \State  Build $G'$, a copy of the residual network $G_{\widehat{f}}$
    % \STATE Let  be a copy of $G_{\widehat{f}}$
    \State  Add super-source $s^*$, edge $(s^*, s)$ with capacity $\eta$ to $G'$
    \State  Let $f^{\textsf{init}}$ saturate $(s^*,s)$ and have flow 0 on all other edges
    \State  Run Algorithm \ref{alg:ws_pr} with inputs $G'$ and $f^{\textsf{init}}$, call output $f'$
    \State  Delete $f'_{(s^*,s)}$ from $f'$
    \State  {\bfseries Output}: $f = f'+\widehat{f}$
\end{algorithmic}
\end{algorithm}
\fi

We create the auxiliary graph $G'$ as in Algorithm \ref{alg:pre-process},
and then run the gap-relabeling Push-Relabel on $G'$ (together with the standard initializing pre-flow) to find a minimum cut between $s^*$ and $t$ and obtain a flow $f'$. Corollary \ref{cor:cs_pr} bounds the Push-Relabel run-time in this case. Adding $f'$ to $\widehat{f}$ creates a cut-saturating pseudo-flow.

The next lemma proves the output of this algorithm satisfies the desired properties and that the algorithm runs in time $O(\eta \cdot n^2)$.
\begin{lemma}
\label{lem:fixing_cut}
Suppose $\widehat{f}$ is a predicted pseudo-flow with error $\eta$ for network $G$.
Then Algorithm \ref{alg:pre-process} finds a cut-saturating pseudo-flow $f$ for $G$ 
with error $ \eta$ in time $O(\eta \cdot n^2)$.
\end{lemma}
\begin{proof}
\iffalse
To begin, if the prediction $\widehat{f}$
violates capacity constraints, 
we round down the flow on violated edges.
More formally, for any edge $e \in E$ with  $\widehat{f}_e> c_e$, 
update the predicted flow to be $\widehat{f}_e= c_e$. 
This step just transfers all of the error in the prediction to be 
on the conservation constraints, 
and not on the capacity constraints; it can be done without loss of generality in time $O(|E|)$.
\fi

In the residual graph $G_{\widehat{f}}$, 
the min-cut is bounded by $\eta$, since it is at most $\eta$ far from being cut-saturating. 
Therefore, we can apply Corollary \ref{cor:cs_pr} to $G_{\widehat{f}}$ 
and obtain an optimal flow $f'$ on $G_{\widehat{f}}$ in $O(\eta \cdot n^2)$ running time.

The flow we desire is $f_e=f_e'+\widehat{f}_e$ for all $e \in E$. 
It is cut-saturating for $G$ by the optimality of $f'$ on $G_{\widehat{f}}$. 
Further, it is a pseudo-flow since $f'$ does not have any excess or deficit in $G_{\widehat{f}}$ and clearly $f'_e+\widehat{f}_e \leq c_e$ for all $e \in E$.
\end{proof}

Notably, one can also run  Algorithm \ref{alg:ws_pr} and terminate it upon finding the min-cut, in which case $f'$ will be a pre-flow on $G_{\widehat{f}}$, and the resulting $f=f' + \widehat{f}$ will have total excess bounded by $2\eta$. In fact, one can do this in other steps of the algorithm as well, if the goal is only to find a min-cut, 
and only lose an additional constant factor in the running time; see Appendix \ref{sec:early-termination}.
Additionally, in practice one may wish to use a predicted cut instead of finding a cut-saturating pseudo-flow as in Algorithm \ref{alg:pre-process}; see the discussion in Appendix \ref{sec:predict-cut}.

\subsection{Saturating a cut separating excesses from deficits}
\label{sec:flip}

Once we have a cut-saturating pseudo-flow $f$, 
which by Lemma \ref{lem:fixing_cut} can be obtained from the prediction 
using Algorithm \ref{alg:pre-process}, 
we are ready to define the accompanying heights and cut using Algorithm \ref{alg:initialheights} again.
Note that the initial cut with two sides $T_0 = \{u\in V: \exists  u-t \text{ path in }G_{f}\}$ and $S_0 = V \setminus T_0$ 
is by definition the same cut as that induced by the heights (as in Lemma \ref{lem:induced-cut}).

We update the pseudo-flow so that it always maintains a saturated cut, but eventually, the nodes with excess and the nodes with deficit are separated by the saturated cut. This is a generalization of what happens with Algorithm \ref{alg:ws_pr}, where we transfer all excess nodes to the $s$-side of the cut. Here, we transfer all excess to the $s$-side, and all deficit to the $t$-side of the cut. 
Interestingly, we observe that this is the sufficient condition for the pseudo-flow to saturate a min-cut.

\begin{lemma}
\label{lem:no-excess-T}
For a cut-saturating pseudo-flow $f$ for a network $G$, let $\delta(S,T)$ be a cut it saturates.
If all the nodes in $T$ have no excess and all the nodes in $S$ have no deficit, then the cut is a minimum cut. 
\end{lemma}
Lemma \ref{lem:no-excess-T} is essentially the analog of Lemma \ref{lem:pre_flow_no_excess_T} in the more general pseudo-flow setting. The proof techniques are similar---we prove that a flow can be found by sending all excess flow back to $s$, and by sending flow from $t$ to all deficits.
This fixes all excess and deficit, while maintaining the same cut. 

To prove Lemma \ref{lem:no-excess-T}, we use the following result from Davies et al. \cite{davies23b}:
\begin{lemma}[Lemma 5, restated from \cite{davies23b}]
\label{lem:old_lemma}
Given any pseudo-flow $f$ for network $G$, every excess node has a path in $G_f$ \emph{to} either a deficit node or $s$; every deficit node has a path in $G_f$ \emph{from} either an excess node or $t$. 
\end{lemma}
\begin{proof}[Proof of Lemma \ref{lem:no-excess-T}]
Consider the residual network $G_f$.
By Lemma \ref{lem:old_lemma}, every excess node $u$ in $S$ must have a path to either a deficit node or to $s$. Since the current pseudo-flow $f$ saturates a cut, the path cannot go across this cut and reach $T$, where all the deficits are. Therefore, $u$ has a path back to $s$, which only uses nodes in $S$. 
Similarly, by Lemma \ref{lem:old_lemma}, for every deficit node $v \in T$ there is a path that starts with either an excess node or $t$ and ends with $v$. Again, all excesses are in $S$ and the cut $\delta(S,T)$ is already saturated by $f$, so there is no path from $S$ to $T$. This path then is from $t$ to $v$ and only uses nodes in $T$.

It follows that we can send all excess to $s$ and send flow from $t$ to all deficit nodes until the pseudo-flow becomes a feasible flow. Notice that $\delta(S,T)$ remains saturated in this process. A feasible flow saturating a cut is a max-flow, and $\delta(S,T)$ is a min-cut.
\end{proof}

By Lemma \ref{lem:no-excess-T}, it is sufficient to find a pseudo-flow and accompanying saturated cut where the excess nodes are all on the $s$-side and the deficit nodes are all on the $t$-side. We begin by focusing on the nodes on the $t$-side of the cut, then briefly justify that 
the same can be done for the $s$-side by considering the backwards network.

\paragraph{Moving excess to the $s$-side.} To resolve all excess on the $t$-side, 
we solve an auxiliary graph problem, where the goal is to send the maximum amount of flow from excess nodes to either deficit nodes or $t$ \emph{within} the $t$-side (currently denoted $T_0$). If the max-flow in this problem matches the total excess in $T_0$, all excess can be resolved locally and only deficits remain; otherwise, the max-flow solution on the auxiliary graph also provides us with a min-cut that ``blocks'' excess nodes from deficit nodes and $t$. This cut will become the new cut maintained by the pseudo-flow after adding the auxiliary flow to it.

For the construction of the auxiliary graph $G'$, take the residual graph induced on the nodes $T_0$, $G_f[T_0]$. Add a super-source and -sink $s^*$ and $t^*$ to it. 
Add edges $(s^*,u)$ with capacity $\textsf{exc}_f(u)$ for every excess node $u \in T_0$; add edges $(v,t^*)$ with capacity $\textsf{def}_f(v)$ for every deficit node $v \in T_0$; and finally, add an edge $(t, t^*)$ with capacity $\eta+1$.

\begin{algorithm}[t!] 
\caption{Moving all excess to the $s$-side of the cut} \label{alg:separate-exc-def}
\begin{algorithmic}
\State  {\bfseries Input}: Network $G$, a cut saturating pseudo-flow $f$
\State  Run Algorithm \ref{alg:initialheights}, get output heights $h$
\State  Let $T_0= \{u\in V: \exists  u-t \text{ path in }G_{f}\}$ and $S_0 = V \setminus T_0$
\State  Build the residual $G_{f}$
    \State  Build $G'$ on copy of $G_{f}[T_0]$ plus $ \{s^*,t^*\}$
    \For{excess node $u \in T_0 \setminus \{t\}$ } 
    \State  Add edge $(s^*,u)$ with capacity $\textsf{exc}_f(u)$
    \EndFor
    \For{deficit node $v \in T_0$ } 
    \State  Add edge $(v,t^*)$ with capacity $\textsf{def}_f(u)$
    \EndFor
    \State  Add edge $(t,t^*)$ with capacity $\eta+1$
    \State  Let $f^{\textsf{init}}_{(s^*,u)}=c_{(s^*,u)}$ for all $(s^*,u)$, and all other $f^{\textsf{init}}_e=0$
\State  Run Algorithm \ref{alg:ws_pr} on $G'$ and $f^{\textsf{init}}$, outputs $f'$ and $T_0'$, $T_0''$
\For{all copies of $e=(u,v) \in E(G_f)$ where $f'_e>0$}
    \State  Update $f_e \leftarrow f_e+f'_e$
\EndFor
    \State  {\bfseries Output}: Flow $f$ and cut parts $S_0 \cup T_0'$ and $T_0''$
\end{algorithmic}
\end{algorithm}
 
When we run cold-start Push-Relabel (Algorithm \ref{alg:ws_pr}) on 
$G'$,
it outputs a flow $f'$ and the $s^*-t^*$ cut $\delta(T'_0, T''_0)$.
Note that $t \in T''_0$, since $(t, t^*)$ has infinite capacity and therefore cannot be in the cut.
Any $s^*-t^*$ path $p$ in $G'$ along which $f$ sends $\delta$ units of flow exactly identifies nodes $u$ and $v$ (where $(s^*,u) \in p$ and $(v,t^*) \in p$) for which $\delta$ units of flow can be sent from $u$ to $v$ along the interior of $p$ in $G_f$. Thus we can send flow as indicated by $f'$ to update $f$. See Algorithm \ref{alg:separate-exc-def} for details.
We obtain the following guarantee on the updated pseudo-flow $f$.
\begin{claim}
\label{claim:help-excess-t}
In Algorithm \ref{alg:separate-exc-def}, the output pseudo-flow $f$ saturates the cut $\delta(S_0 \cup T'_0, T''_0)$, and all excess nodes are in $S \cup T'_0$. 
Moreover, the total excess and deficit in $G$ has not increased.
\end{claim}
\begin{proof}
Let $f_{\textsf{old}}$ denote the input to Algorithm \ref{alg:separate-exc-def}.

The fact that the output $f$ saturates the cut 
$\delta(S \cup T'_0, T''_0)$ immediately follows from the fact 
that $f'$ saturated the cut $\delta(T'_0, T''_0)$ in $G'$. Indeed, all edges from $T'_0$ to $T''_0$ are now saturated and all edges from $T''_0$ to $T'_0$ have no flow. All edges from $S$ to $T''_0$ are already saturated in the old flow $f_{\textsf{old}}$ and remain so after adding $f'$ since its flows are locally within $T_0$. For the same reason, all edges from $T''_0$ back to $S$ still have no flow.
 
Now, we consider the total excess and deficit.
First note that the nodes that have excess/deficit with respect to the 
updated pseudo-flow $f$ are a subset of the nodes 
that had excess/deficit with respect to $f_{\textsf{old}}$, and
the excess/deficit of a node is clearly never increased.

Assume for sake of contradiction that there is an excess node $u \in T_0''$.
Then $u$ had excess with respect to $f_{\textsf{old}}$ too, 
so there is an edge $(s^*,u)$ that had capacity $\textsf{exc}_{f_{\textsf{old}}}(u)$ in $G'$ but was not saturated by $f'$.
Further, since a min-cut in $G'$ is $\delta( T'_0, T''_0)$, 
it must be that $u $ can reach $t$ in $G'$. 
This means that in $G'$ there is a path with positive remaining capacity  between $s^*$ and $t^*$, 
contradicting the fact that $f'$ was a max-flow in $G'$.
\end{proof}

By Claim \ref{claim:help-excess-t}, 
the updated $f$ satisfies the conditions of Lemma \ref{lem: elim-exc-def} 
by taking $S^* = S_0 \cup T_0'$ and $T^* = T_0''$.
Lastly, the run-time claimed in Lemma \ref{lem: elim-exc-def} follows by
applying Corollary \ref{cor:cs_pr} on $G'$.
Putting everything together we have the following lemma.

\begin{lemma}
\label{lem: elim-exc-def}
Let $f$ be a pseudo-flow for network $G$ with error $\eta$ 
that saturates cut $\delta(S_0,T_0)$.
Algorithm \ref{alg:separate-exc-def} finds a new cut-saturating pseudo-flow in time $O(\eta \cdot n^2)$ so that the new pseudo-flow saturates an additional $s-t$ cut $\delta(S^*,T^*)$ that has no excess nodes in $T^*$, and the total excess and deficit is still bounded by $\eta$. 
\end{lemma}

\paragraph{Moving deficits to the $t$-side.} Next, we will do a similar procedure for the $s$-side of the cut, 
though this time we wish to remove deficit nodes. We will show that this is exactly the backward process of what happens to the $t$-side, and can be done by reversing the graph edges and flows and running Algorithm \ref{alg:separate-exc-def} on the reversed network.

  We will build the reverse network of $G$, call it $B$ (for backwards).
    The network $B$ consists of a copy of $G$ but all of the edges go the opposite direction. 
    More specifically, for every node $u \in V(G)$ there is a mirror node $u'$ in $B$, and for every edge $e=(u,v) \in E(G)$ with capacity $c_e$, there is a mirror edge $e'=(v', u') \in E(B)$ with capacity $c_e$.
    Note that the source $s$ in $G$ is mirrored to the sink $s'$ in $B$, whereas the sink $t$ in $G$ is mirrored to the source $t'$ in $B$.

    We can reverse any pseudo-flow $f$ on $G$ to be another pseudo-flow $f'$ on $B$, where for all $e \in E(G) $, $f'_{e'} = f_e$.
    Notably, $f$ and $f'$ saturate the same cut, and we observe
    $\textsf{exc}_{f}(u) = \textsf{def}_{f'}(u')$ and $\textsf{def}_{f}(u) = \textsf{exc}_{f'}(u').$

Suppose we have a pseudo-flow $f$ that saturates cut $\delta(S_0, T_0)$ in $G$ with no excess nodes in $T_0$.  
Then in the backwards network $B$, $f'$ saturates $\delta(T'_0, S'_0)$, where $T_0$ (resp. $S_0$) is all mirror nodes $p'$ for such $p \in T_0$ (resp. $S_0$). Now $S'_0$ becomes the sink-side of the cut.  
In $B$, we can send flow from excess nodes and $s'$ to deficit nodes within $S'_0$, and this can be done 
by running Algorithm \ref{alg:separate-exc-def} on $B$.

The true algorithm for $G$ is Algorithm \ref{alg:separate-exc-def_two}, 
which defer to Appendix \ref{sec:mirror}, since it is really just the mirror image of Algorithm \ref{alg:separate-exc-def}, 
though we may skip the execution of Algorithm \ref{alg:initialheights}, as we already know the cut.

This flow, when reversed back into $G$, is the maximum amount of flow that can go from excess nodes and $s$ to deficit nodes in $G_f[S_0]$. 
After adding this reversed flow to $f$, 
    the result is a cut-saturating pseudo-flow for $G$, where there is no deficit on the $s$-side of the cut.
    Observe that there is no excess or deficit created on either side of the cut in the process.

We obtain the following corollary of Lemma \ref{lem: elim-exc-def}.
\begin{corollary}
\label{cor: elim-s}
Let $f$ be a pseudo-flow for network $G$ with error $\eta$ 
that saturates cut $\delta(S_0,T_0)$.
One can update $f$ in time $O(\eta \cdot n^2)$ so that all flow in $T_0$ remains unchanged, but now $f$ saturates a cut $\delta(S^*,T^*)$ and there are no deficit nodes in $S^*$.
\end{corollary}

\subsection{From min-cut to max-flow}

Summarizing this section, we prove our main theorem.
\begin{proof}[Proof of Theorem \ref{thm:main}]
Given a predicted pseudo-flow $\widehat{f}$ with error $\eta$ on network $G$, 
Lemma \ref{lem:fixing_cut} proved that Algorithm \ref{alg:pre-process} 
finds a cut-saturating pseudo-flow $f$ for $G$ 
with error $ \eta$ in time $O(\eta \cdot n^2)$.
To find a min-cut, 
Lemma \ref{lem:no-excess-T} shows that it is enough to find a pseudo-flow saturating a cut so that the $t$-side of the cut contains no excess and the $s$-side of the cut contains no deficit.

We run Algorithm \ref{alg:separate-exc-def} seeded with $f$ on $G$
to obtain an updated cut-saturating pseudo-flow with no excess on the $t$-side of the maintained cut by Lemma \ref{lem: elim-exc-def}.
Then, Algorithm \ref{alg:separate-exc-def} can be run on the backwards network $B$, 
and from Corollary \ref{cor: elim-s}, the updated 
cut-saturating pseudo-flow now has no excess on the $t$-side of the cut and no deficit on the $s$-side.

The last phase of the algorithm can be left out if only the min-cut is desired; suppose the min-cut is $\delta(S, T)$.
% , for example in applications such as image segmentation. 
% However, we still show this for completeness of the algorithm. 
By the proof of Lemma \ref{lem:no-excess-T}, to obtain a max-flow we only need to send all excess flow back to $s$, and send flow from $t$ to every deficit node. Label all nodes in $S$ with height $n$ and all nodes in $T$ with height $0$.
Then run Algorithm \ref{alg:ws_pr} to fix all excess in $S$. 
The algorithm will only send flow back to $s$, since there is no way to cross the cut $\delta(S, T)$. Then reverse the graph and flow, and again run Algorithm \ref{alg:ws_pr} to fix the excess nodes in the reversed graph, which exactly correspond to the deficit nodes in the original graph.
\end{proof}

\section{Empirical Results}
\label{sec: empirical}
In this section, we validate the theoretical results in Sections \ref{sec:ws_pseudoflow}. To demonstrate the effectiveness of our methods, we consider \emph{image segmentation}, a core problem in computer vision that aims at separating an object from the background in a given image. 
It is common practice to re-formulate image segmentation as a max-flow/min-cut optimization problem (see for example \cite{boykov2001interactive, boykov2004experimental, boykov2006graph}), and solve it with combinatorial graph-cut algorithms.

The experiment design we adopt largely resembles that in \cite{davies23b}, which studied warm-starting the Ford-Fulkerson algorithm for max-flow/min-cut. As in previous work, we do not seek state-of-the-art running time results 
for image segmentation.
Our goal is to show that on real-world networks, warm-starting can lead to significant run-time improvements for the Push-Relabel min-cut algorithm, which claims stronger theoretical worst-case guarantees and empirical performance than the Ford-Fulkerson procedures. % in \cite{davies23b}. 
We highlight the following:

\begin{itemize}
    \item 
Our implementation of cold-start Push-Relabel is much faster than Ford-Fulkerson on these graph instances, enabling us to explore the effects of warm-starting on larger image instances. This improved efficiency results from implementing the gap labeling and global labeling heuristics, both known to boost Push-Relabel's performance in practice.
\item As we increase the number of image pixels (i.e., the image's resolution), the size of the constructed graph increases and the savings in time becomes more significant.
\item Implementation choices (such as how to learn the seed-flow from historical graph instances and their solutions) that make the predicted pseudo-flow cut-saturating and that reroute excesses and deficits are crucial to the efficiency of warm-starting Push-Relabel.
\end{itemize}

\begin{figure}
        \centering
        \begin{subfigure}[b]{0.20\linewidth}
            \centering
            \includegraphics[width=\linewidth]{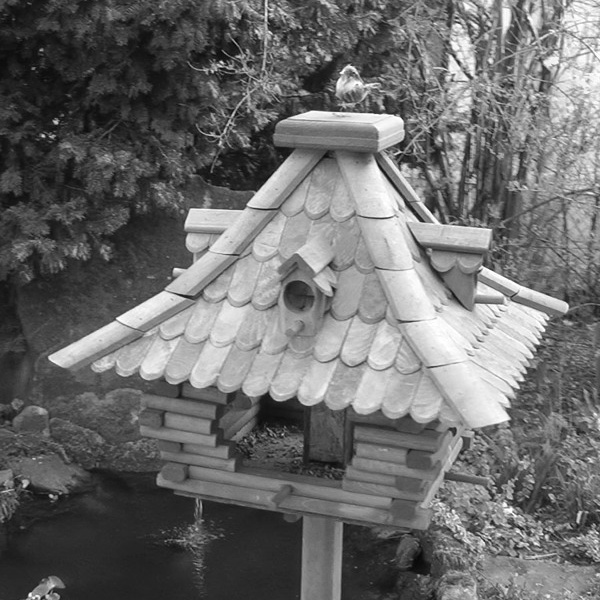}
            \caption[]%
            {{\small Birdhouse}}    
            \label{fig:birdhouse_c}
        \end{subfigure}
        \hfill
        \begin{subfigure}[b]{0.20\linewidth}  
            \centering 
            \includegraphics[width=\linewidth]{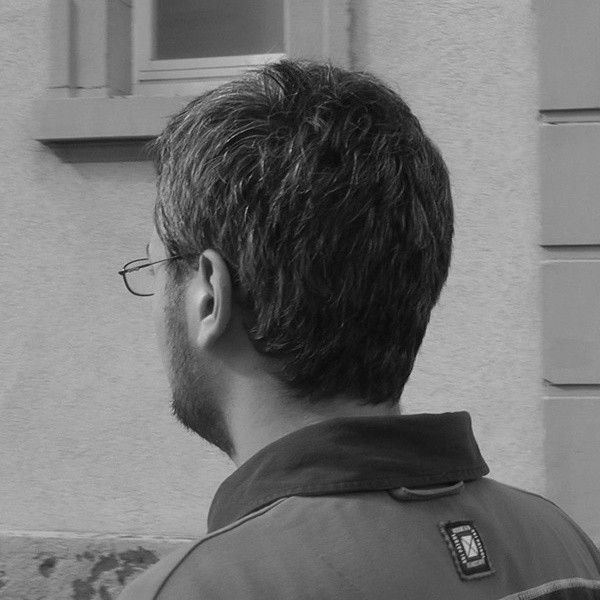}
            \caption[]%
            {{\small Head}}    
            \label{fig:head_c}
        \end{subfigure}
        \hfill
        %\vskip\baselineskip
        \begin{subfigure}[b]{0.20\linewidth}   
            \centering 
            \includegraphics[width=\linewidth]{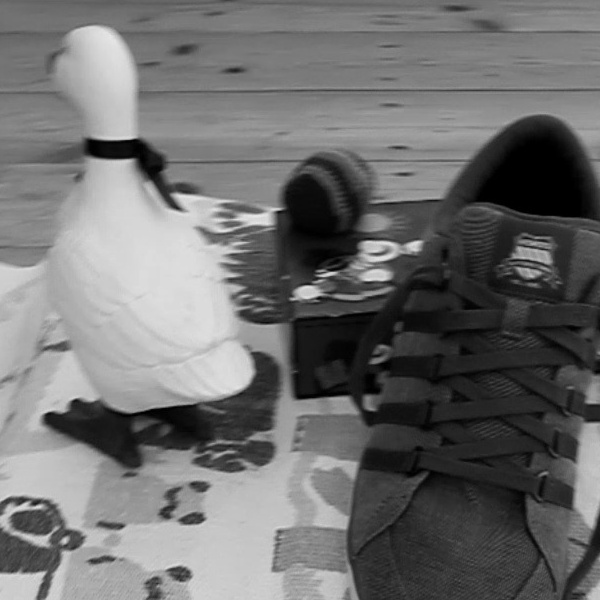}
            \caption[]%
            {{\small Shoe}}    
            \label{fig:shoe_c}
        \end{subfigure}
        \hfill
        \begin{subfigure}[b]{0.20\linewidth}   
            \centering 
            \includegraphics[width=\linewidth]{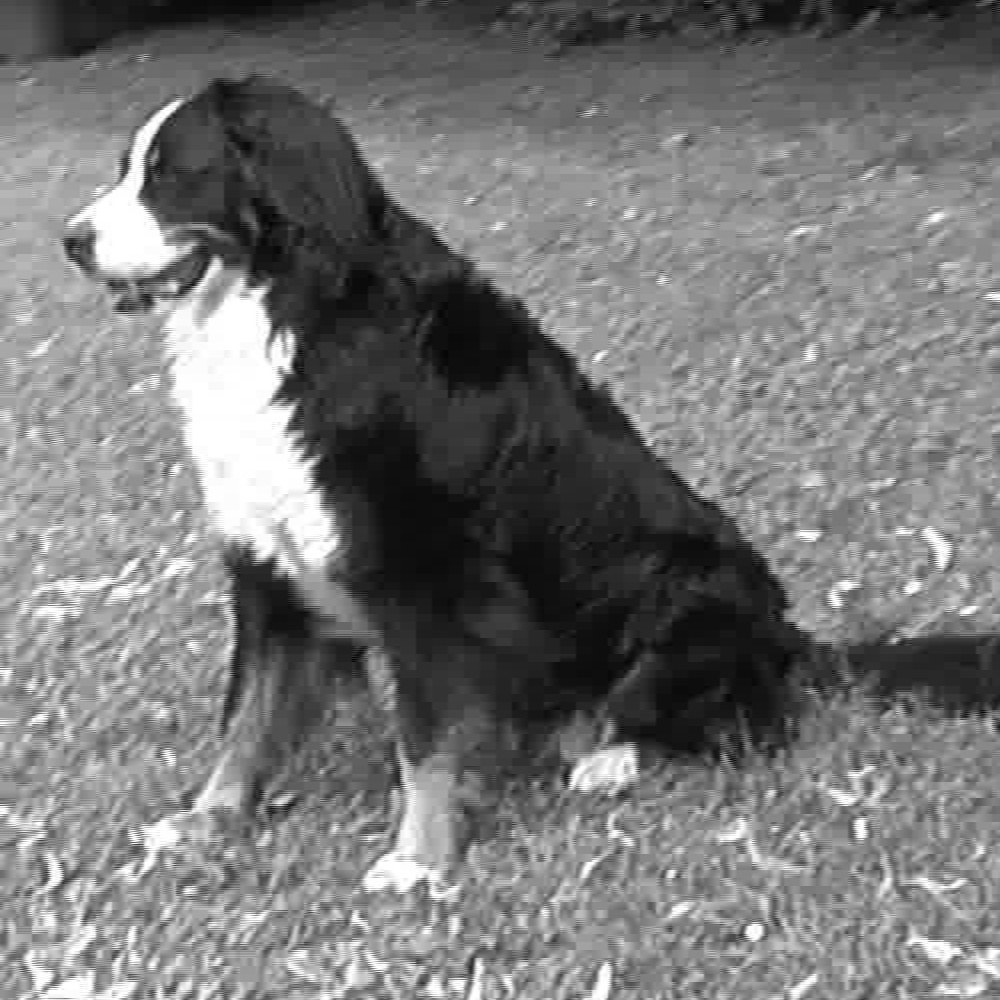}
            \caption[]%
            {{\small Dog}}    
            \label{fig:dog_c}
        \end{subfigure}
        \caption{\small The cropped and gray-scaled images from Figure \ref{fig:original_images} (copy from Figure 2 in \cite{davies23b}).} 
        \label{fig:cropped_images}
% \vspace{-0.2in}
\end{figure}

\paragraph{Datasets and data prepossessing}
Our image groups are from the \emph{Pattern Recognition and Image Processing} dataset from the University of Freiburg, and are titled \textsc{Birdhouse}, \textsc{Head}, \textsc{Shoe}, and \textsc{Dog}.  The first three groups are .jpg images from the \emph{Image Sequences}\footnote{https://lmb.informatik.uni-freiburg.de/resources/datasets/sequences.en.html} dataset. The last group, \textsc{Dog}, was a video that we converted to a sequence of .jpg images from the \emph{Stereo Ego-Motion}\footnote{https://lmb.informatik.uni-freiburg.de/resources/datasets/StereoEgomotion.en.html} dataset. 

Each of the image groups consists of a sequence of photos of an object and its background. 
There are slight variations between consecutive images in a sequence, which are the result of the object and background's relative movements or a change in the camera's position.
These changes alter the solution to the image segmentation problem, 
but the effects should be minor when the change between consecutive images is minor.
In other words, we expect an optimal flow and cut found on an image in a sequence to be a good prediction for the next image in the sequence.

From each group, we consider 10 images and crop them to be either $600 \times 600$ or $500 \times 500$ pixel images, still containing the object, and 
gray-scale all images.
We rescale the cropped, gray-scaled images to be $N \times N$ pixels to produce different sized datasets.
Experiments are performed for $N \in \{30, 60, 120, 240, 480\}$. 
In the constructed graph, we have $|V| = N^2 + 2$. Every graph is sparse, with $|E| = O(|V|)$, hence both $|V|$ and $|E|$ grow as $O(N^2)$.
Detailed description of raw data and example original images can be found in Appendix \ref{sec:exp_app} (Table \ref{table:data_desc}, Figure \ref{fig:original_images}).

\paragraph{Graph construction}
As in \cite{davies23b}, we formulate image segmentation as a max-flow/min-cut problem. The construction of the network flow problem applied in both our work and theirs is derived from a long-established line of work on graph-based image segmentation; see \cite{boykov2006graph}. The construction takes pixels in images to be nodes; and a penalty function value which evaluates the contrast between the pigment of any neighboring pixels to be edge capacity. We leave details on translating the images to graphs on which we solve max-flow/min-cut to Appendix \ref{sec:exp_app}.

\paragraph{Implementation details in warm-start Push-Relabel}
Throughout the experiments,
whenever the Push-Relabel subroutine is called on any auxiliary graph, it is implemented with the gap relabeling heuristic, as shown in Algorithm \ref{alg:ws_pr}, and the \emph{global relabeling} heuristic, which occasionally updates the heights to be 
a node's distance from $t$ in the residual graph.
These heuristics are known to improve the performance of Push-Relabel. As a tie-breaker for choosing the next active node to push from, we choose the one with highest height, which is known to improve the running time of Push-Relabel. We found the generic Push-Relabel algorithm without these heuristics to be slower than Ford-Fulkerson.

All images from the same sequence share the same seed sets.
The constructed graphs are on the same sets of nodes and edges, but the capacities on the edges are different.
The first image in the sequence is solved from scratch. 
For the second image in the sequence, we reuse the old optimal flow and cut from the first image one, 
then for the $i^{th}$ image in the sequence, we reuse the optimal flow and cut from the ${i-1}^{st}$ image. We reuse the old max-flow on the new network by rounding down the flow on edges whose capacity has decreased, hence producing excesses and deficits, and pass this network and flow to the warm-start Push-Relabel algorithm in Section \ref{sec:ws_pseudoflow}. 

\begin{figure}
        \centering
        \begin{subfigure}[b]{0.2\linewidth}
            \centering
            \includegraphics[width=\linewidth]{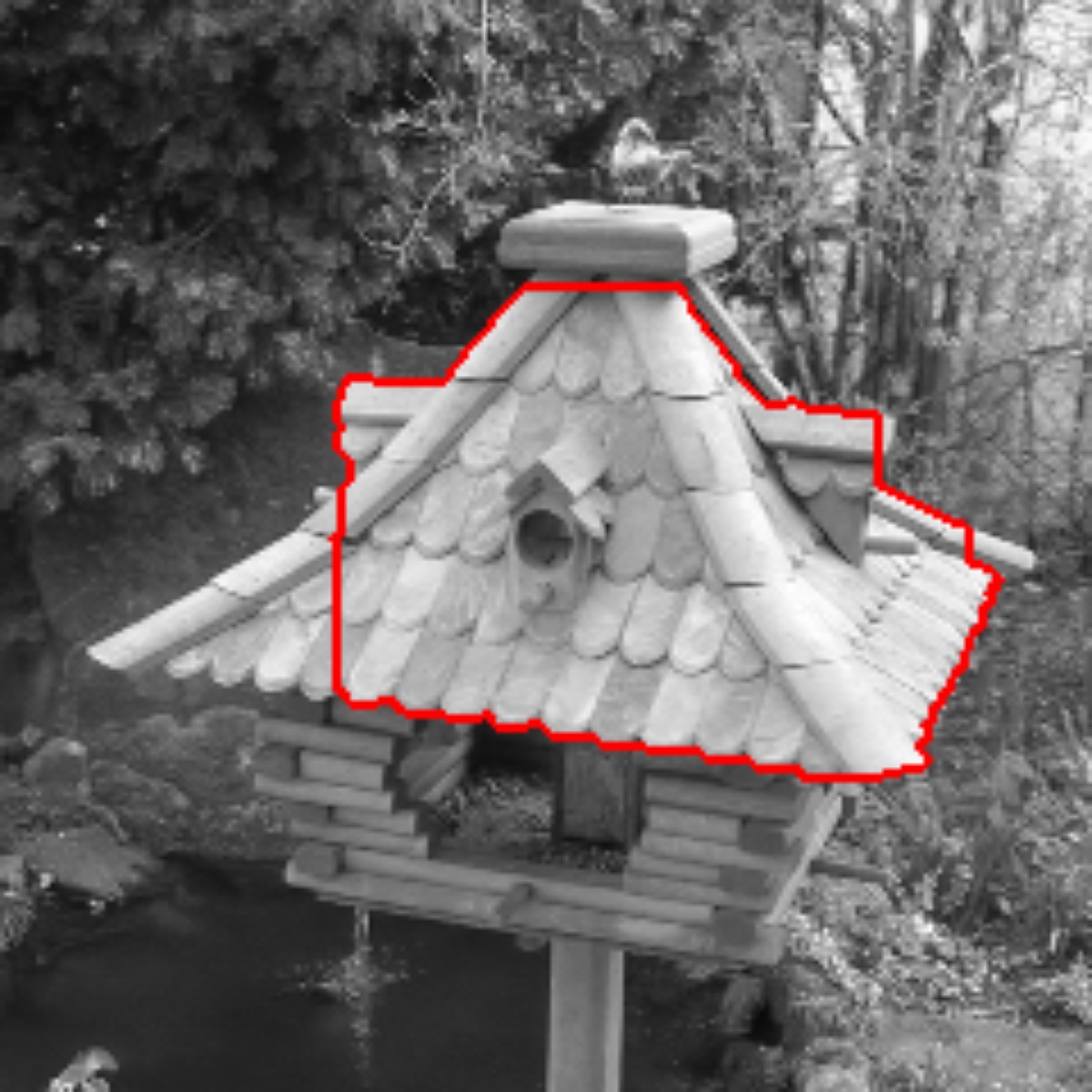}  
            \label{fig:bhcut1}
        \end{subfigure}
        \hspace{1cm}
        \begin{subfigure}[b]{0.2\linewidth}  
            \centering 
            \includegraphics[width=\linewidth]{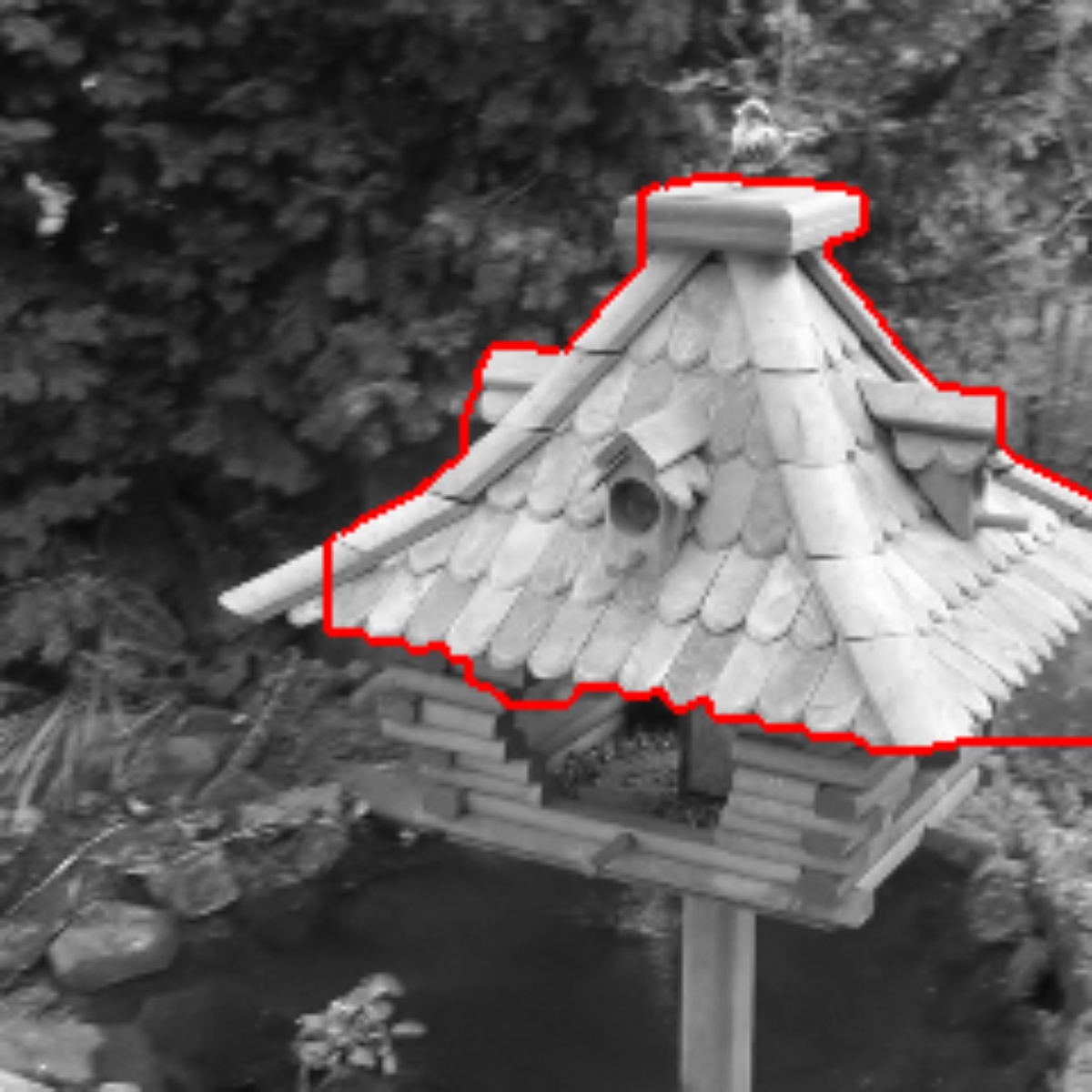}   
            \label{fig:bhcut2}
        \end{subfigure}
       \hspace{1cm}
            \begin{subfigure}[b]{0.2\linewidth}  
            \centering 
            \includegraphics[width=\linewidth]{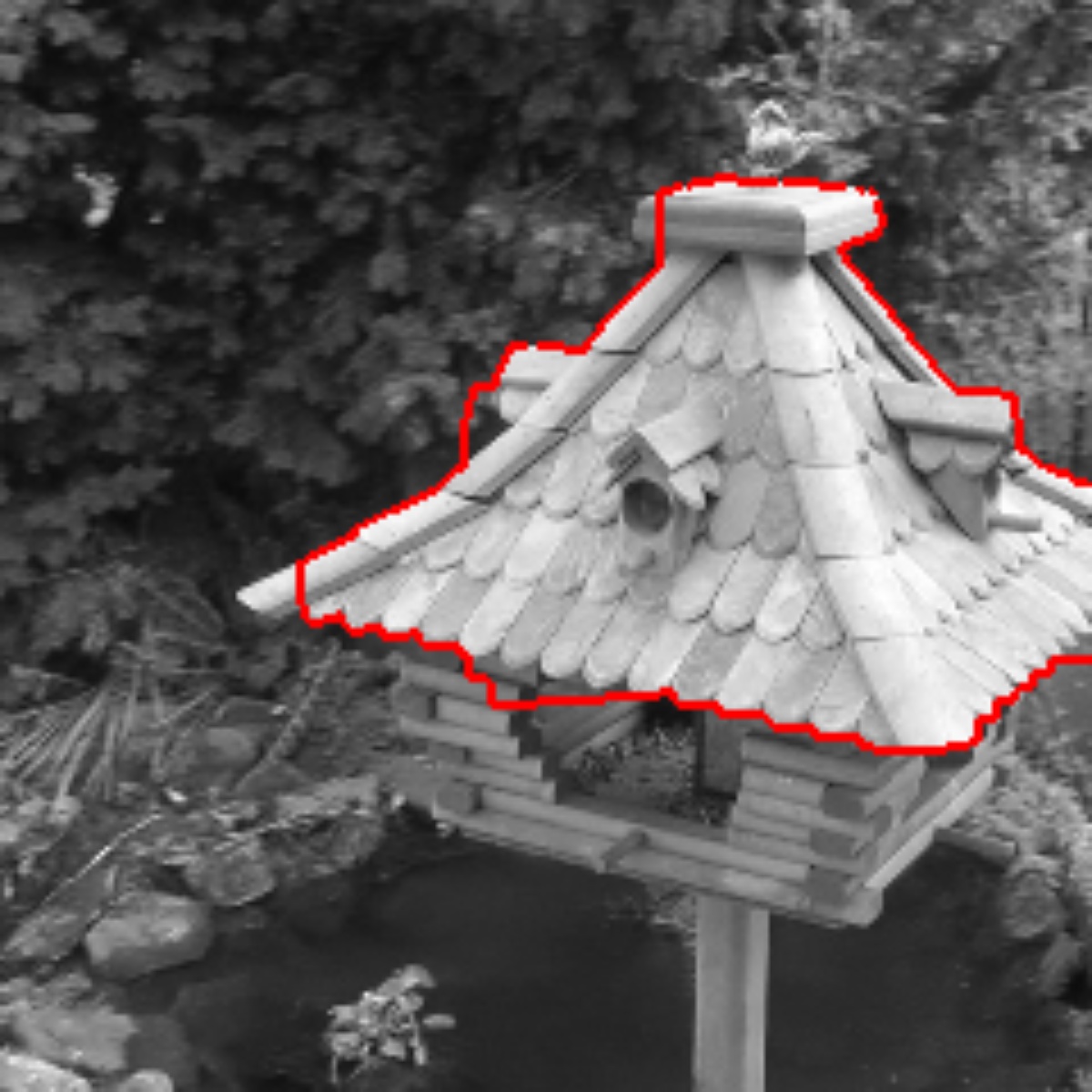}   
            \label{fig:bhcut3}
        \end{subfigure}
        \caption{\small  Cuts (red) on images chronologically evolving from the $240 \times 240$ pixel images from \textsc{Birdhouse}.} 
        \label{fig:bh_cuts}
\end{figure}

To find a saturating cut, instead of sending flow from $s$ to $t$ as suggested in Algorithm \ref{alg:pre-process}, we reuse the min-cut on the previous image $\delta(S_0, T_0)$ and send flow from $S_0$ to $T_0$ that originates from either $s$ or an excess node, and ends at either $t$ or a deficit node. We experimented with a few different ways of projecting the old flow to a cut-saturating one on the new graph. The way we implemented was by far the most effective, although it shares the same theoretical run-time as Algorithm \ref{alg:pre-process}. 

The graph-based image segmentation method finds reasonable object/background boundaries. Figure \ref{fig:bh_cuts} shows an example of how the target cut could evolve as the image sequence proceeds. Even with the same set of seeds, the subtle difference in images could lead to different min-cuts that need to be rectified. However, the hope is that the old min-cut bears much resemblance to the new one, hence warm-starting Push-Relabel with it could be beneficial. See Appendix \ref{sec:exp_app} for other examples.

\noindent \textbf{Results}\quad Table \ref{table:running_time_ff_vs_pr_partial} shows average running times for both Ford-Fulkerson in \cite{davies23b} and Push-Relabel for two image sizes: $120 \times 120$ (the largest size tested in prior work) and $480 \times 480$. The full data on all data sizes are in Table \ref{table:running_time_ff_vs_pr_full} in Appendix \ref{sec:exp_app}. Table \ref{table:running_time_scale} shows how the run-time of warm-start Ford-Fulkerson and Push-Relabel scales with growing image sizes on image group \textsc{Dog}. The ``N/A'' in both tables marks overly long run-time ($>$1 hour), at which point we stop evaluating the exact run-time. The run-time growth across data sizes on other image groups can be found in Appendix \ref{sec:exp_app}.

\begin{table}[ht]
\centering
\caption{Average run-times (s) of cold-/warm-start Ford Fulkerson (FF) and Push-Relabel (PR)}
\label{table:running_time_ff_vs_pr_partial}
\begin{tabular}{r|ccccl}
\hline
    Image Group & FF cold-start & FF warm-start  &  PR cold-start & PR warm-start  \\ 
\hline
\textsc{Birdhouse} $120\times120$ & 109.06  &  37.31 &  5.42 & 4.98 \\
    \textsc{Head} $120\times120$& 101.79 & 28.43  & 5.90& 5.92 \\
    \textsc{Shoe} $120\times120$& 98.95 & 30.44  & 6.44 & 3.74\\
    \textsc{Dog} $120\times120$ & 190.36 & 38.08   & 6.76 & 6.38 \\ 
\hline
\hline
\textsc{Birdhouse} $480\times480$ & N/A  &  N/A &  604.54 & 502.58 \\
\textsc{Head} $480\times480$ & N/A  &  N/A &  365.25 & 285.75 \\
\textsc{Shoe} $480\times480$ & N/A  &  N/A &  756.77 & 364.42 \\
\textsc{Dog} $480\times480$ & N/A  &  N/A &  834.63 & 363.41 \\
\hline
\end{tabular}
\end{table}

\begin{table}[h!]
\centering
\caption{Growth of average running times of warm-start Ford Fulkerson (FF) and Push-Relabel (PR) in seconds, on image group \textsc{Dog}}
\label{table:running_time_scale}
\begin{tabular}{r|ccccc}
\hline
    Algorithm & $30\times30$ & $60\times60$ & $120 \times 120$ & $240 \times 240$ & $480 \times 480$ \\
\hline
    Ford-Fulkerson & 0.41 & 6.89 & 42.04 & 459.48 & NA \\
    Push-Relabel & 0.11 & 0.95 & 6.38 & 52.42 & 363.41 \\
\hline
\end{tabular}
\end{table}

These results show warm-starting Push-Relabel, while slightly losing in efficiency on small images, greatly improves in it on large ones. As for the scaling of run-time with growing data sizes, both cold- and warm- start's running time increases polynomially with the image width $n$, but warm-start scales better, and as $n$ increases to $480$, it gains a significant advantage over cold-start. Despite the different warm-start theoretical bounds ($O(\eta |V|^2)$ for Push-Relabel versus $O(\eta |E|)$ for Ford-Fulkerson), in practice both warm-start algorithms scale similarly as the dataset size grows, as shown in Table \ref{table:running_time_scale}.

Table \ref{table:pr_break_down} shows how the running time of warm-start Push-Relabel breaks down into the three phases described in Section \ref{sec:ws_pseudoflow}: (1) finding a cut-saturating pseudo-flow; (2) fixing excess on $t$-side; (3) fixing deficits on $s$-side. Note phase (1) takes the most time, but results in a high-quality pseudo-flow, in that it takes little time to fix the excess/deficits appearing on the ``wrong'' side of the cut. 

% \vspace{-0.1in}
\begin{table}[ht]
\centering
\caption{Running time of warm-start Push-Relabel break down, on \textsc{Birdhouse}}
\label{table:pr_break_down}
% \tiny
\begin{tabular}{r|cccccl}
\hline
    Size & $30 \times 30$ & $60 \times 60$ & $120\times120$ & $240 \times 240$ & $480 \times 480$ \\ 
\hline
    Total & 0.06 & 0.45 & 4.98 & 55.68 & 502.58 \\ 
    Saturating cut & 0.04 & 0.34 & 4.17 & 46.25 & 431.49\\
    Fixing $t$ excesses & 0.01 & 0.09 & 0.53 & 5.29 & 64.01\\
    Fixing $s$ deficits & 0.01 & 0.02 & 0.27 & 4.13 & 7.08\\
\hline
\end{tabular}
% \vspace{-0.1in}
\end{table}

\section{Conclusions}
We provide the first theoretical guarantees on warm-starting Push-Relabel with a predicted flow, improving the run-time from $O(m \cdot n^2)$ to $O(\eta \cdot n^2)$.  Our algorithm uses one the most well-known heuristics in practice, the gap relabeling heuristic, to keep track of cuts in a way that allows for provable run-time improvements.

One direction of future work is extending the approaches in this work to generalizations of $s$-$t$ flow problems, for instance, tackling minimum cost flow or multi-commodity flow. An ambitious goal of such an agenda would be to develop new warm-start methods for solving arbitrary linear programs.  

A different line of work is to develop rigorous guarantees for other empirically proven heuristics by analyzing them through a lens of predictions, providing new theoretical insights and developing new algorithms for fundamental problems.

%%%%%%%%%%%%%%%%%%%%%%%%%%%%%%%%%%%%%%%%%%%%%%%%%%%%%%%%%%%%%%%%%%%%%%%%%%%%%%%
%%%%%%%%%%%%%%%%%%%%%%%%%%%%%%%%%%%%%%%%%%%%%%%%%%%%%%%%%%%%%%%%%%%%%%%%%%%%%%%
% APPENDIX
%%%%%%%%%%%%%%%%%%%%%%%%%%%%%%%%%%%%%%%%%%%%%%%%%%%%%%%%%%%%%%%%%%%%%%%%%%%%%%%
%%%%%%%%%%%%%%%%%%%%%%%%%%%%%%%%%%%%%%%%%%%%%%%%%%%%%%%%%%%%%%%%%%%%%%%%%%%%%%%
\newpage 

\printbibliography

\newpage
\appendix
\onecolumn

\section{More Discussion on Warm-starting Push-Relabel}
We include some more insights and details on warm-starting Push-Relabel.

\subsection{Tackling Unknown $\eta$ Value}\label{sec:predict-cut}

 Notice that, Algorithm \ref{alg:pre-process} directly uses the error value $\eta$, and treats it as given input. Due to the definition of $\eta$, it covers both the total excess/deficit and how far the current flow $\hat{f}$ is from being cut-saturating. The former is easy to measure; whereas the latter is not, since there are numerous cuts in $G$ and it is not obvious which one is closest to being saturated by $\hat{f}$. However, this challenge can be tackled by running the algorithm in the binary-search fashion. 
Say there is an $\eta^*$ which is the true error from being cut-saturating. We cannot know this value for sure because we cannot compute the value on each possible cut. One can start with some very small value of $\eta$ (such as $1$), put it on the edge $(s^*, s)$, and try to use Push-Relabel to send the flow from $s$ to $t$. If we successfully send the current $\eta$ from $s$ to $t$, there exists a $s-t$ flow of $\eta$ in the residual graph; meaning when Push-Relabel terminates the $s^*-t$ cut we will find is just the edge $\eta$. If this is the case, double $\eta$, saturate $(s^*, s)$, and again use Push-Relabel to send the extra $\eta$ downstream to $t$. Repeat this until the returned cut is not $(s^*, s)$. It will then be a $s-t$ cut. Intuitively, the role $\eta$ plays here is just a surplus of flow provided to the source $s$; hence it should not be bounding the flow-sending. Otherwise it means the current pseudo-flow is not yet cut-saturating. This gives the same run-time bound as each time we double $\eta$, the excess to resolve only increases by $\eta$; hence the total excess we have resolved throughout all iterations is still $O(\eta^*)$.

It is noteworthy that in experiments, we initialize $\eta$ to be the error computed on the old min-cut on the previous image. While this cut is not necessarily the one that bounds $\eta$, we found it to be an effective surrogate value for the real underlying $\eta$.

Notably, one can also run Algorithm \ref{alg:ws_pr} and terminate it upon finding the min-cut, in which case $f'$ will be a pre-flow on $G_{\hat{f}}$, and the resulting $f=f' + \hat{f}$ will have total excess bounded by $2\eta$. In fact, one can do this in other steps of the algorithm as well, if the goal is only to find a min-cut, 
and only lose an additional constant factor in the running time; see Appendix \ref{sec:early-termination}.
As discussed, in practice one may wish to use a predicted cut instead of finding a cut-saturating pseudo-flow as in Algorithm \ref{alg:pre-process}.

By Definition \ref{def:error}, if a pseudo-flow $\hat{f}$ is $\sigma$ far from cut-saturating it means augmenting it by another flow $f$ with value at most $\sigma$ can saturate some cut. Let this cut be $\delta(S,T)$. Another way to look at this is, within $\hat{f}$, the total flow passing through the cut $\delta(S,T)$ satisfies:
$$\sum_{u \in S, v \in T}\hat{f}_{(u,v)} - \sum_{u \in S, v \in T} \hat{f}_{(v,u)} \geq \sum_{u \in S, v \in T}c_{(u,v)} - \sum_{u \in S, v \in T}c_{(v,u)} + \sigma.$$

Apart from solving max-flow in the residual graph to saturate this cut, there may be other options to create a cut-saturating pseudo-flow. For example, the $\eta$ bound on error does not directly tell us where this cut is. However, if a practitioner can ``guess'' a good enough cut $\delta(S,T)$ from past problem instances, such a pseudo-flow can also be obtained simply by saturating all edges $(u, v) \in \delta(S,T)$ and removing the flow on all backward edges. The downside is that such a practice will transfer the error on that particular cut to the total excess and deficit on nodes incident to the cut. Overall, there may be a trade-off where one can omit Algorithm \ref{alg:pre-process} in lieu of using a predicted cut, but at the cost of having to fix more excess and deficit in later steps.

\subsection{The mirror algorithm}\label{sec:mirror}

\begin{algorithm}[h!] 
\caption{Moving all deficit to the $t$-side of the cut} \label{alg:separate-exc-def_two}
\begin{algorithmic}
\State  {\bfseries Input}: Network $G$, a pseudo-flow $f$ saturating cut $\delta(S_0, T_0)$
\State  Build the residual $G_{f}$.
    \State  Build $G'$ on copy of $G_{f}[S_0]$ plus $\{s^*,t^*\}$
    \For{excess node $u \in S_0$ } 
    \State  Add edge $(s^*,u)$ with capacity $\textsf{exc}_f(u)$
    \EndFor
    \For{deficit node $v \in S_0 \setminus \{s\}$ } 
    \State  Add edge $(v,t^*)$ with capacity $\textsf{def}_f(u)$
    \EndFor
        \State  Add edge $(s^*, s)$ with capacity $\eta+1$ (or sufficiently large capacity)
     \State Let $f^{\textsf{init}}_{(s^*,u)}=c_{(s^*,u)}$ for all $(s^*,u)$ and $f^{\textsf{init}}_{(s^*,s)}=c_{(s^*,s)}$, and all other $f^{\textsf{init}}_e=0$ 
\State  Run Algorithm \ref{alg:ws_pr} on $G'$ and $f^{\textsf{init}}$, outputs $f'$ and $S_0'$, $S_0''$
\For{all copies of $e=(u,v) \in E(G_f)$ where $f'_e>0$}
    \State  Update $f_e \leftarrow f_e+f'_e$
\EndFor
    \State  {\bfseries Output}: Flow $f$ and cut parts $S_0'$ and $S_0'' \cup T_0$
\end{algorithmic}
\end{algorithm}

\subsection{Early termination of auxiliary Push-Relabel upon finding min-cut}\label{sec:early-termination}
In Section \ref{sec:ws_pseudoflow}, we mentioned that one can choose to quit the Push-Relabel algorithm on auxiliary graphs whenever a cut is found. The resulting pseudo-flow, although violating flow conservation constraints, can still be added to the initial pseudo-flow. We give a brief analysis of how this effects the execution of the algorithm.

The pseudo-flow is constructed in three places: 
\begin{enumerate}
\item In Algorithm \ref{alg:pre-process}, where we saturate a cut; 
\item In Algorithm \ref{alg:separate-exc-def}, where we push flow from $t$-side excess nodes to deficit nodes and $t$;
\item In Algorithm \ref{alg:separate-exc-def_two}, where we push flow from $s$-side excess nodes and $s$ to deficit nodes.
\end{enumerate}

Notice this simple fact:
\begin{claim}
For pseudo-flows $f, f', \text{and } f''$ where $f = f' + f''$ (without violating capacities constraints), we have:
$$\sum_{p}(\textsf{exc}_{f}(u) + \textsf{def}_{f}(u)) \leq \sum_{u}(\textsf{exc}_{f'}(u) + \textsf{def}_{f'}(u)) + \sum_{u}(\textsf{exc}_{f''}(u) + \textsf{def}_{f''}(u))$$
\end{claim}
In Step 1, Algorithm \ref{alg:ws_pr} starts with $f_{\textsf{init}}$ with excess $\eta$, hence the resulting pre-flow also has at most $\eta$ excess, and adding this pre-flow without restoring it to a max-flow may increase the excess by $\eta$. In Step 2, the initial flow in $G'$ also has total excess of at most $\sum_{u \in T_0} \textsf{exc}_{\widehat{f}}(u) \leq \eta$, so at the end of Algorithm \ref{alg:separate-exc-def} the total excesses also increases by this much. In Step 3, correspondingly the maximum increase is $\sum_{u \in S_0} \textsf{def}_{\widehat{f}}(u) \leq \eta$. To sum up, early termination in the auxiliary networks after finding the min-cut increases the total error by $O(\eta)$, and therefore has the same run-time bound up to a constant factor.

\section{More on Experiments}
\label{sec:exp_app}

\subsection{More on graph construction}

We take as input an image on pixel set $V$, and two sets of \emph{seeds} $\mathcal{O}, \mathcal{B} \subseteq V$. 
The seed set $\mathcal{O}$ contains pixels that are known to be part of the object, while the seed set $\mathcal{B}$ contains pixels that are known to be part of the background.
The \emph{intensity} or gray scale of pixel $v$ is denoted by $I_v$.
We say that two pixels are neighbors if they are either in the same column and in adjacent rows 
or same row and adjacent columns.
Intuitively, if neighboring pixels have very different intensities, we might expect one to be part of the object and one to be part of the background.
For any two pixels $p,q \in V$, a solution that separates them, i.e., puts one pixel in the object and the other one in the background, incurs a \emph{penalty} of $\beta_{p, q}$.
For neighbors $p$ and $q$, $\beta_{p,q}=C \exp (-(I_p - I_q)^2/(2\sigma^2))$, for $C$ a large constant, otherwise the penalty is $0$. Note that the quantity $\beta_{p,q}$ gets bigger when neighbors $p$ and $q$ have stronger contrast. 

A segmentation solution seeded with $\mathcal{O}$ and $\mathcal{B}$ labels each pixel as either being part of the object or part of the background, and the labeling must be consistent with the seed sets.
Let $J$ denote the object pixels for a fixed segmentation solution.
Then the \emph{boundary-based} objective function is the sum of all of the  penalties
$\max_{J}  \sum_{p \in J, q \notin J}\beta_{p,q},$ for $J$ with $\mathcal{O} \subseteq J, \mathcal{B} \subseteq V \setminus J $.
As in the definition, a positive penalty cost is only incurred on the object's boundary. The goal is to minimize the total penalty, 
which is in turn
maximizing the contrast between the object and background, for the given object and background seed sets.

Solving this maximization problem is
equivalent to solving the max-flow/min-cut problem on the following network. 
There is a node for each pixel, plus the object terminal $s$ and the background terminal $t$. 
As notation suggests, $s$ is the source of the network and $t$ is the sink.
The edge set on the nodes is as follows:
(1) for every $v \in \mathcal{O}$ add edge $(s,v)$ with capacity $M$, for $M$ a huge enough value that it is never saturated in any optimal cut;
(2) for every $u \in \mathcal{B}$ add edge $(u,t)$, again with capacity $M$;
(3) for every pair of nodes $p,q \in V$, add edges $(p, q)$ and $(q, p)$ with capacity $\beta_{p, q}$.
If an image is on $n \times n$ pixels, 
note that the graph is sparse with $|V|=O(n^2)$ nodes and $|E|=O(n^2)$ edges. 

In our experiments, all $\beta_{p, q}$'s are rounded down to the nearest integer, so that capacities are integral.
Since $\beta_{p, q} \leq C$ by definition, it suffices for us to let $M=C|V|^2$.

\subsection{Omitted tables and figures for experiments}

Table \ref{table:data_desc} contains a detailed description of each of the four image groups, their original size in the raw dataset, the cropped grey-scaled image size, the foreground/background they feature, etc.

% \vspace{-0.2in}
\begin{table}[ht]
\centering
\caption{Image groups' descriptions (copy of Table 1 from \cite{davies23b})}
% \small
\label{table:data_desc}
\begin{tabular}{r|llll}
\hline
    Image Group & Object, background & Original size & Cropped size \\ 
\hline
    \textsc{Birdhouse} & wood birdhouse, backyard & 1280, 720 & 600, 600 \\
    \textsc{Head} & a person's head, buildings & 1280, 720 & 600, 600 \\
    \textsc{Shoe} & a shoe, floor and other toys & 1280, 720 & 600, 600 \\
    \textsc{Dog} & Bernese Mountain dog, lawn & 1920, 1080 & 500, 500\\
\hline
\end{tabular}
% \vspace{-0.2in}
\end{table}

Figure \ref{fig:original_images} gives one example of raw images from each image group.

\begin{figure}[ht]
        \centering
        \begin{subfigure}[b]{0.25\linewidth}
            \centering
            \includegraphics[width=\linewidth]{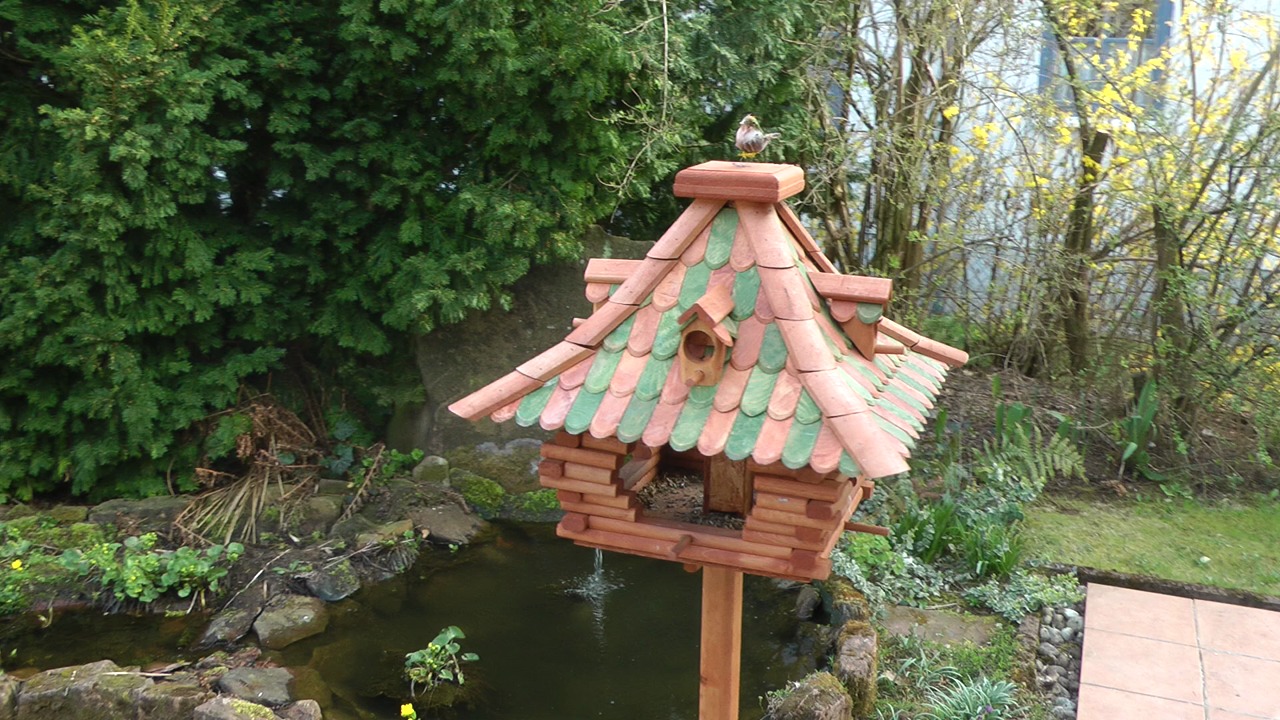}
            \caption[]%
            {{\small Birdhouse}}    
            \label{fig:birdhouse}
        \end{subfigure}
        \hspace{1cm}
        \begin{subfigure}[b]{0.25\linewidth}  
            \centering 
            \includegraphics[width=\linewidth]{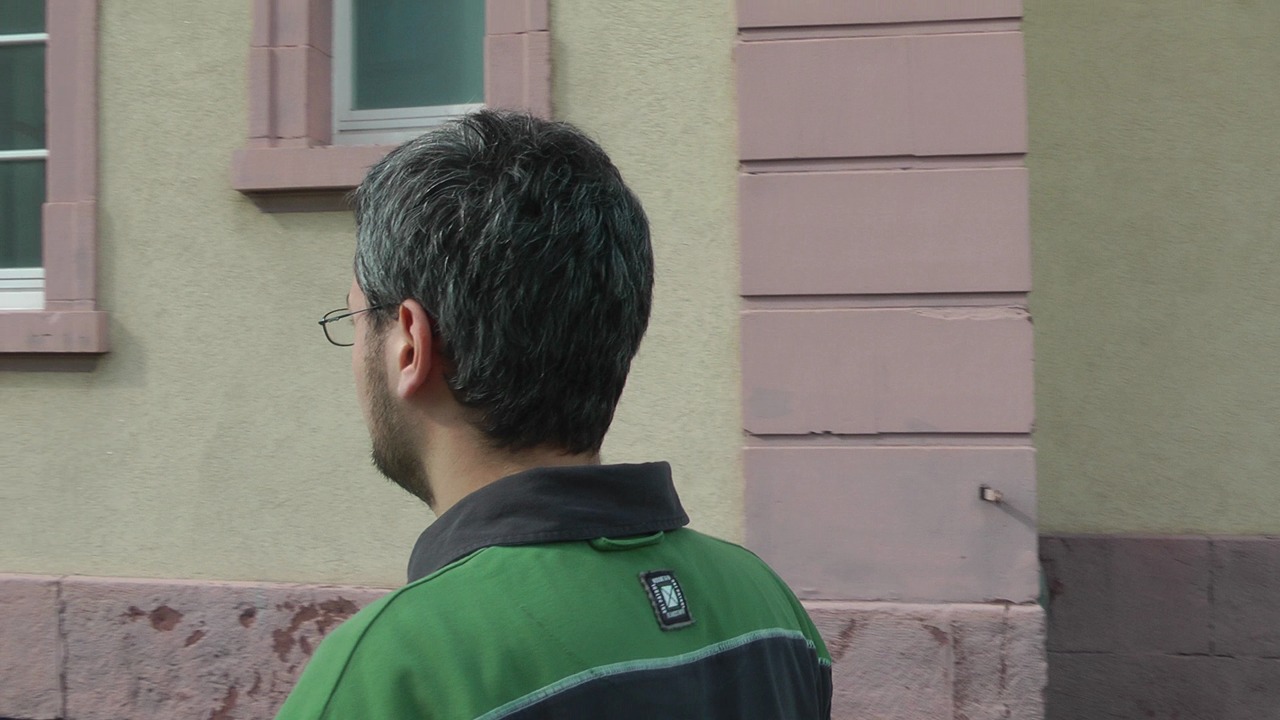}
            \caption[]%
            {{\small Head}}    
            \label{fig:head}
        \end{subfigure}
        \vskip\baselineskip
        \begin{subfigure}[b]{0.25\linewidth}   
            \centering 
            \includegraphics[width=\linewidth]{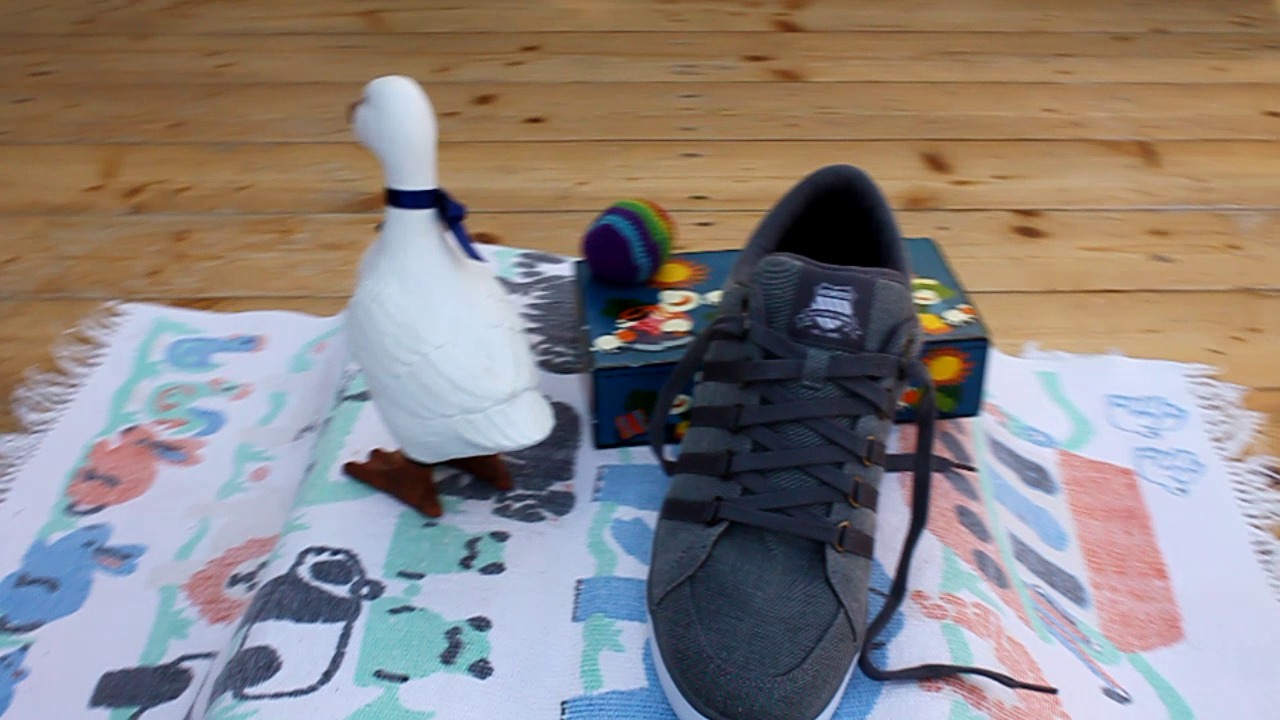}
            \caption[]%
            {{\small Shoe}}    
            \label{fig:shoe}
        \end{subfigure}
        \hspace{1 cm}
        \begin{subfigure}[b]{0.25\linewidth}   
            \centering 
            \includegraphics[width=\linewidth]{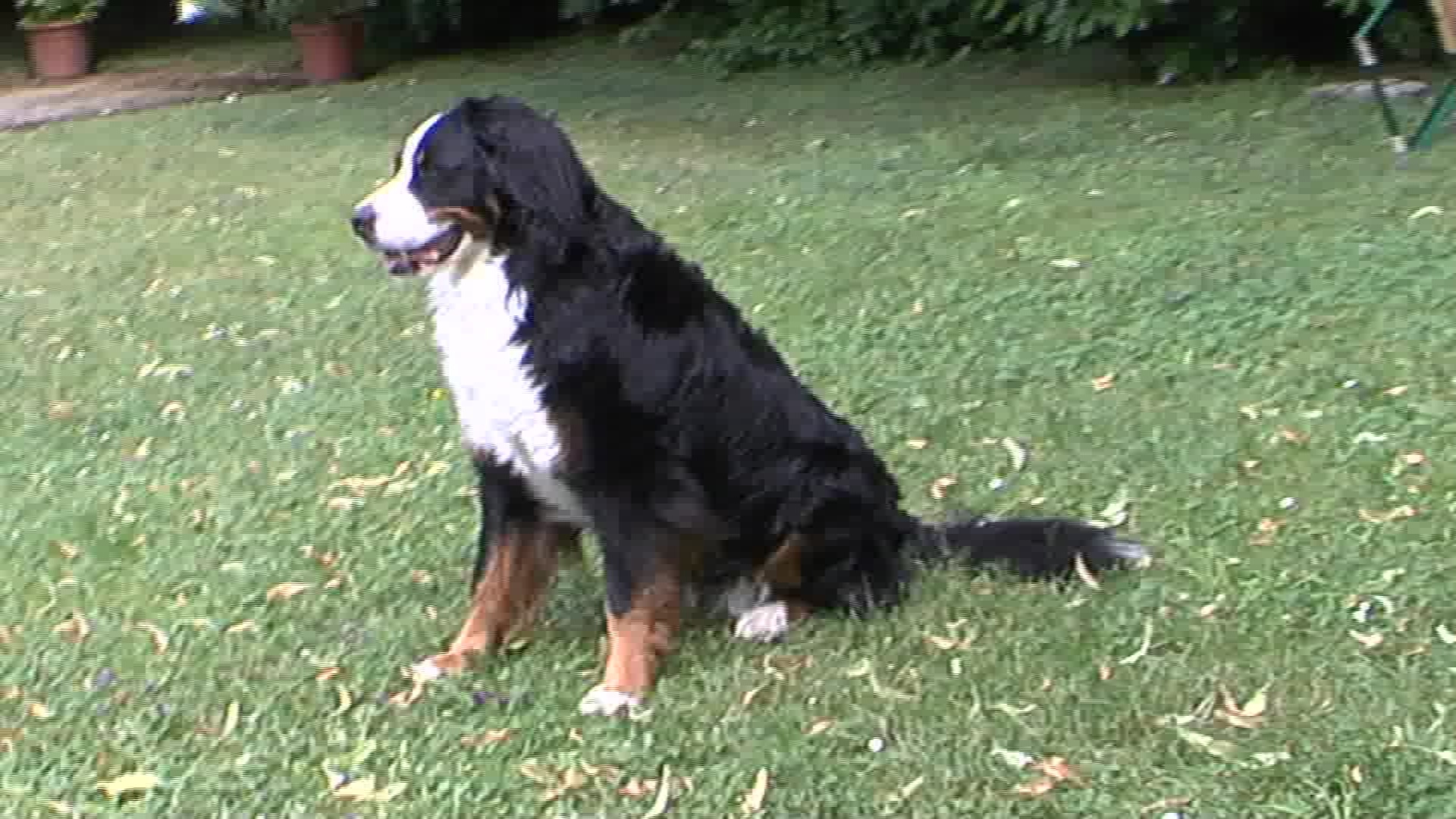}
            \caption[]%
            {{\small Dog}}    
            \label{fig:dog}
        \end{subfigure}
        \caption{\small Instances of images from each group (copy of Figure 1 from \cite{davies23b}).} 
        \label{fig:original_images}
% \vspace{-0.1in}
\end{figure}

In the main body, Figure \ref{fig:bh_cuts} shows examples of cuts found in some images from the \textsc{Birdhouse} image sequence. Figure \ref{fig:dog_cuts}, \ref{fig:shoe_cuts}, \ref{fig:head_cuts} show example cuts from the other image groups. 

\begin{figure}
        \centering
        \begin{subfigure}[b]{0.2\linewidth}
            \centering
            \includegraphics[width=\linewidth]{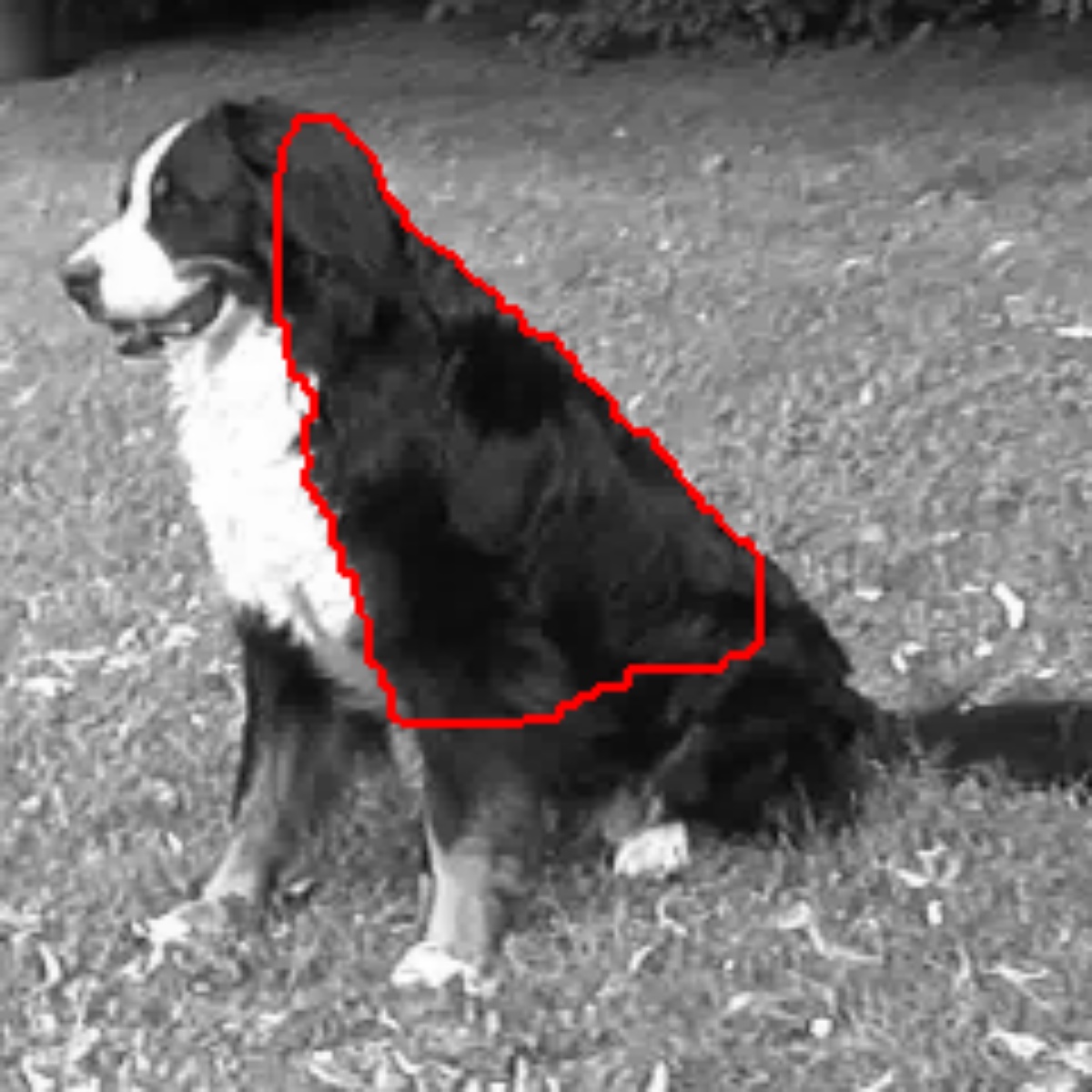}  
            \label{fig:dogcut1}
        \end{subfigure}
        \hspace{1cm}
        \begin{subfigure}[b]{0.2\linewidth}  
            \centering 
            \includegraphics[width=\linewidth]{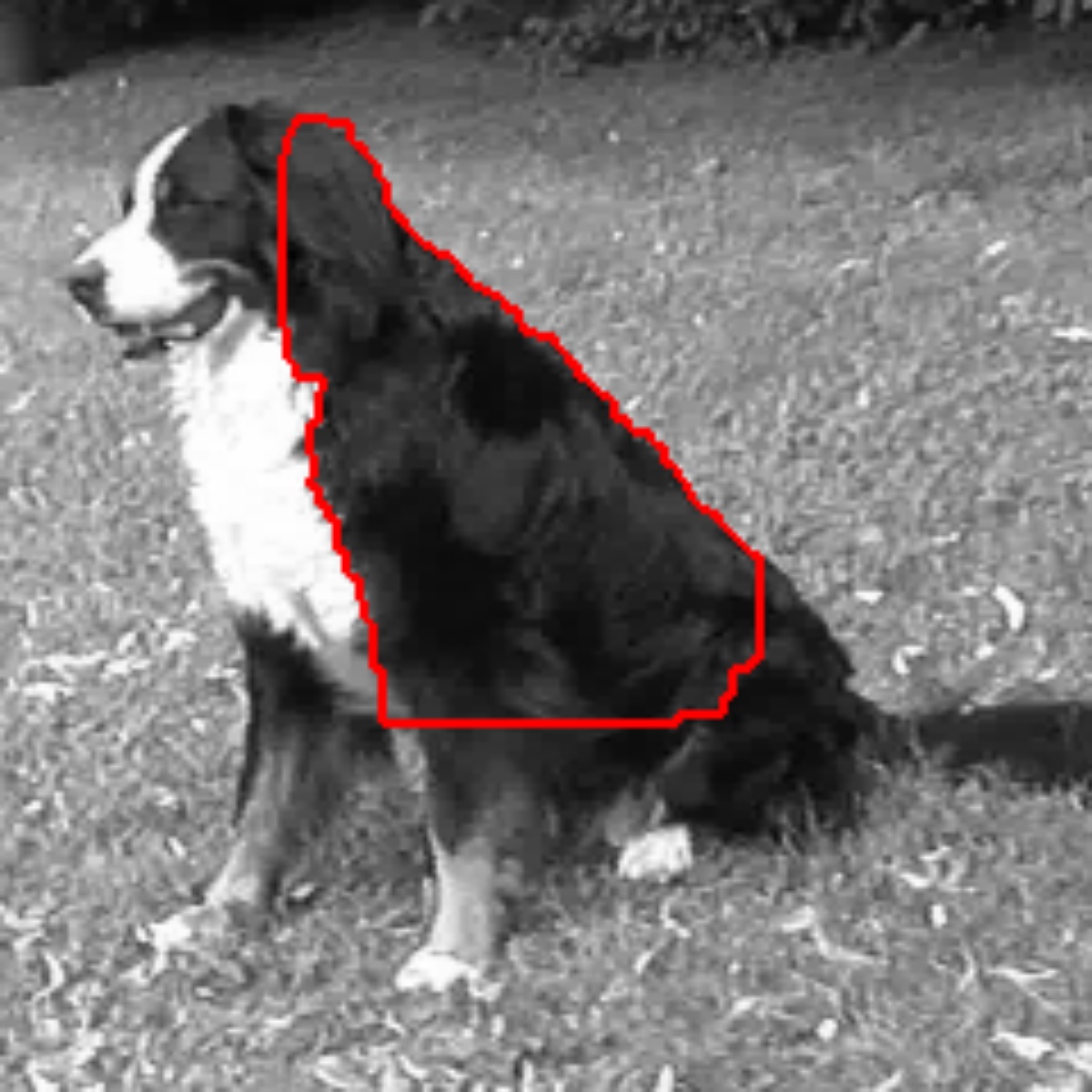}   
            \label{fig:dogcut2}
        \end{subfigure}
       \hspace{1cm}
            \begin{subfigure}[b]{0.2\linewidth}  
            \centering 
            \includegraphics[width=\linewidth]{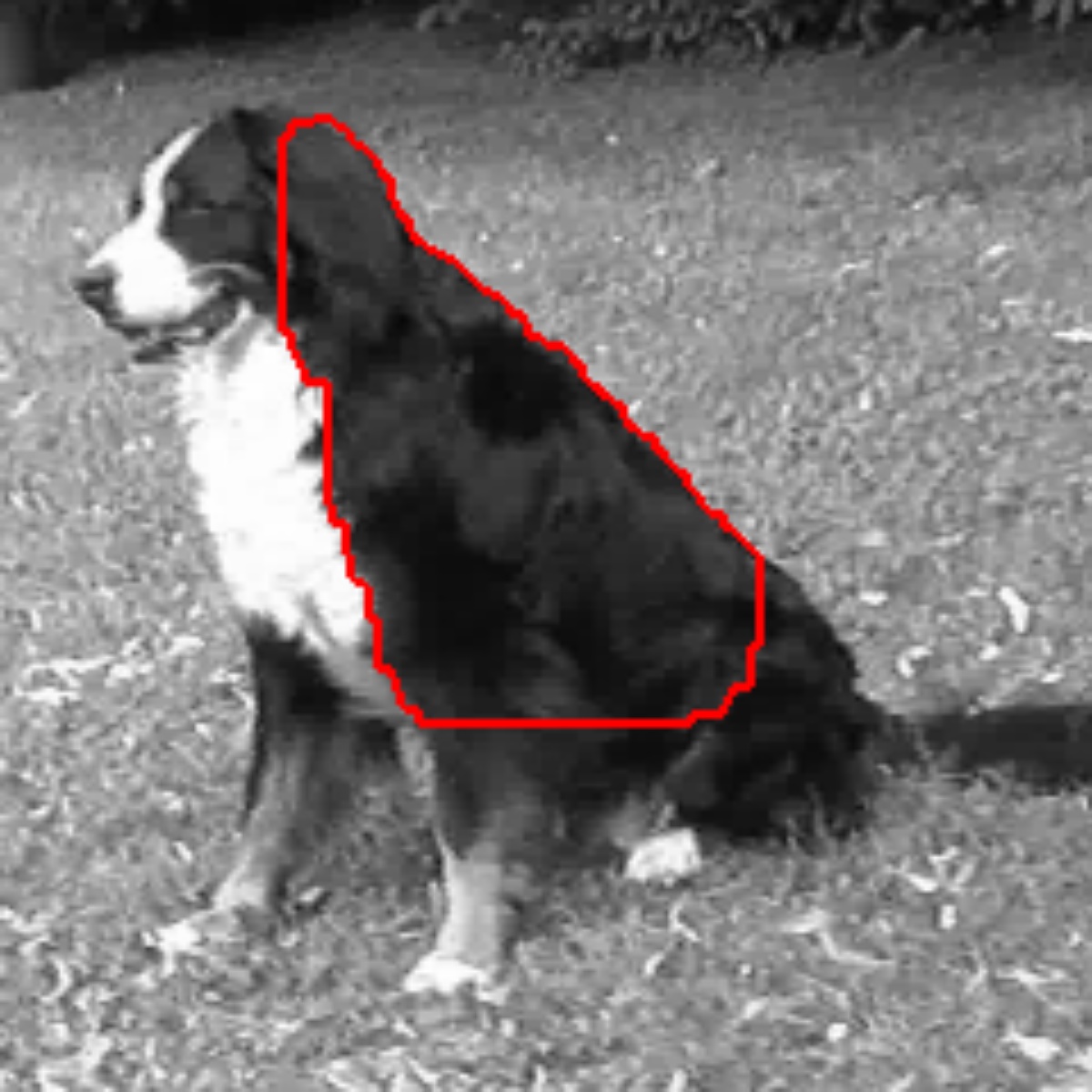}   
            \label{fig:dogcut3}
        \end{subfigure}
        \caption{\small  Cuts (red) on images chronologically evolving from the $240 \times 240$ pixel images from \textsc{Dog}.} 
        \label{fig:dog_cuts}
% \vspace{-0.2in}
\end{figure}

\begin{figure}
        \centering
        \begin{subfigure}[b]{0.2\linewidth}
            \centering
            \includegraphics[width=\linewidth]{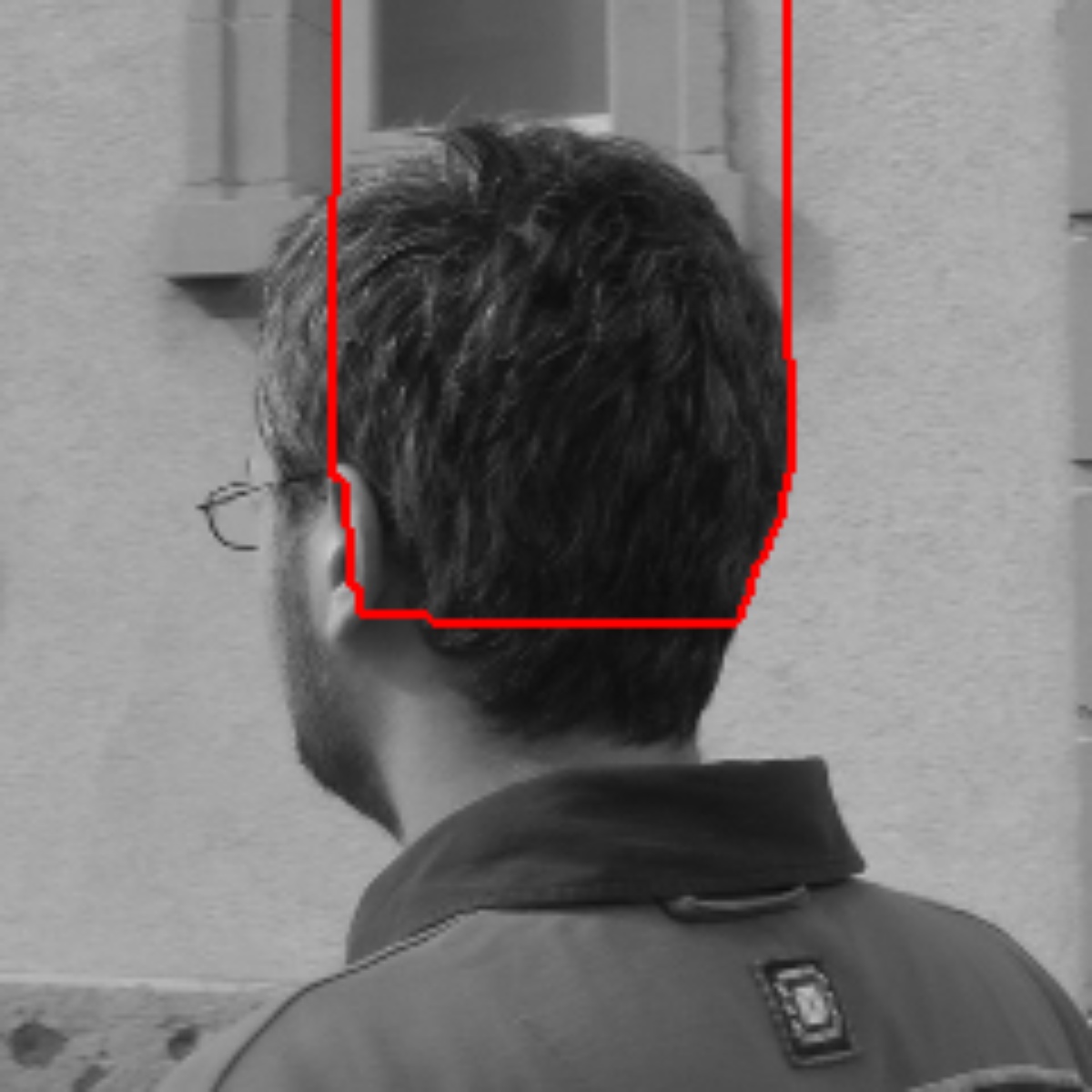}  
            \label{fig:headcut1}
        \end{subfigure}
        \hspace{1cm}
        \begin{subfigure}[b]{0.2\linewidth}  
            \centering 
            \includegraphics[width=\linewidth]{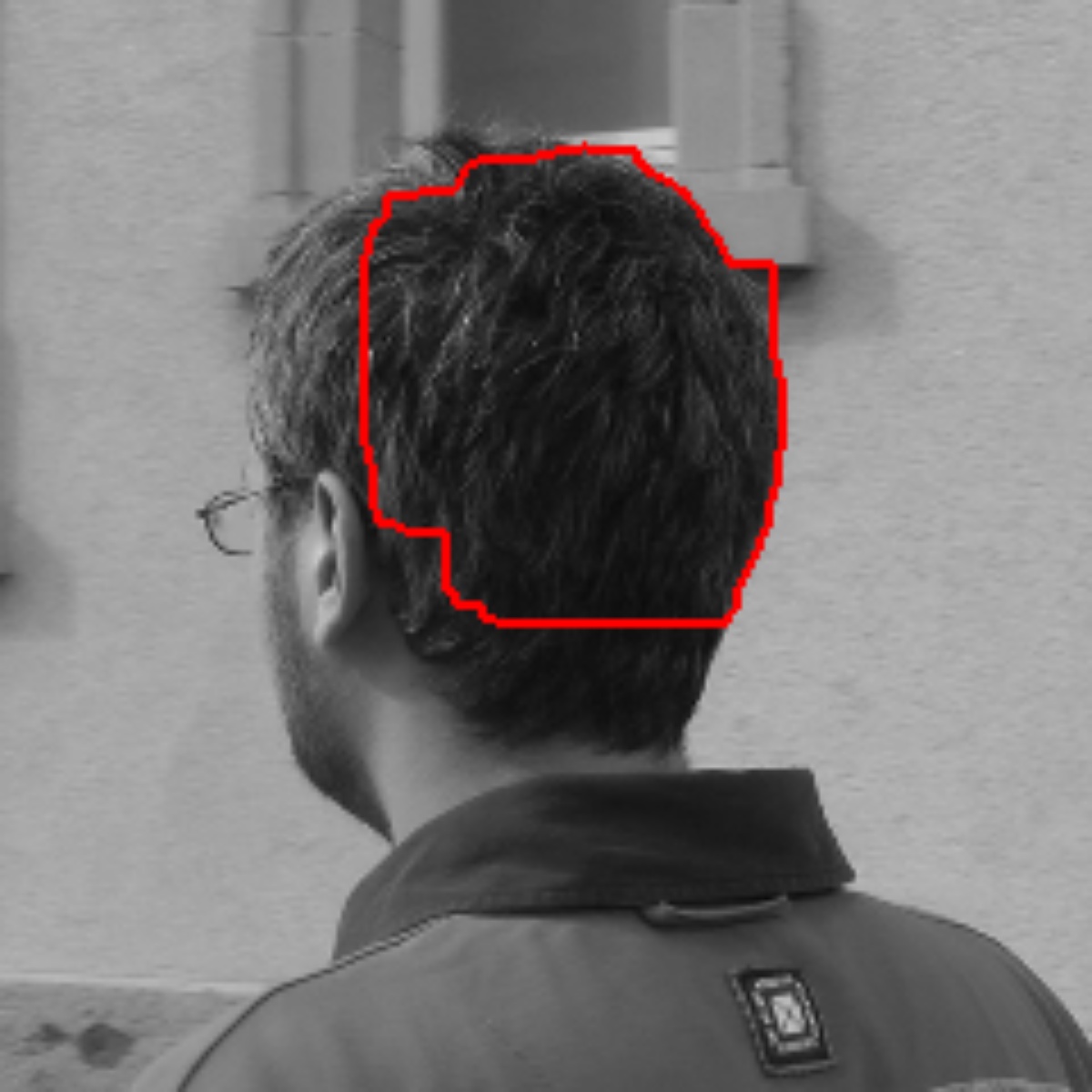}   
            \label{fig:headcut2}
        \end{subfigure}
       \hspace{1cm}
            \begin{subfigure}[b]{0.2\linewidth}  
            \centering 
            \includegraphics[width=\linewidth]{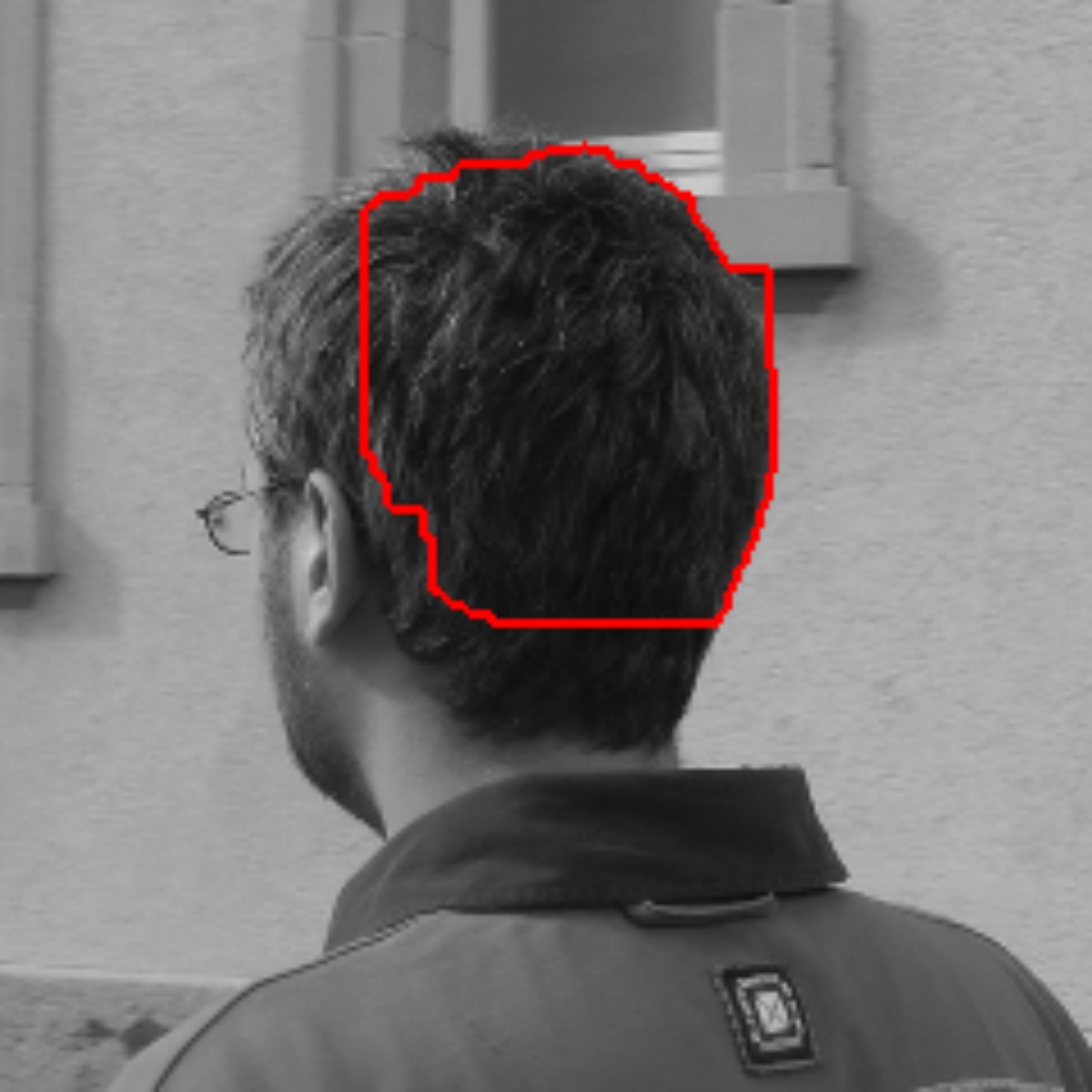}   
            \label{fig:headcut3}
        \end{subfigure}
        \caption{\small  Cuts (red) on images chronologically evolving from the $240 \times 240$ pixel images from \textsc{Head}.} 
        \label{fig:head_cuts}
% \vspace{-0.2in}
\end{figure}

\begin{figure}
        \centering
        \begin{subfigure}[b]{0.2\linewidth}
            \centering
            \includegraphics[width=\linewidth]{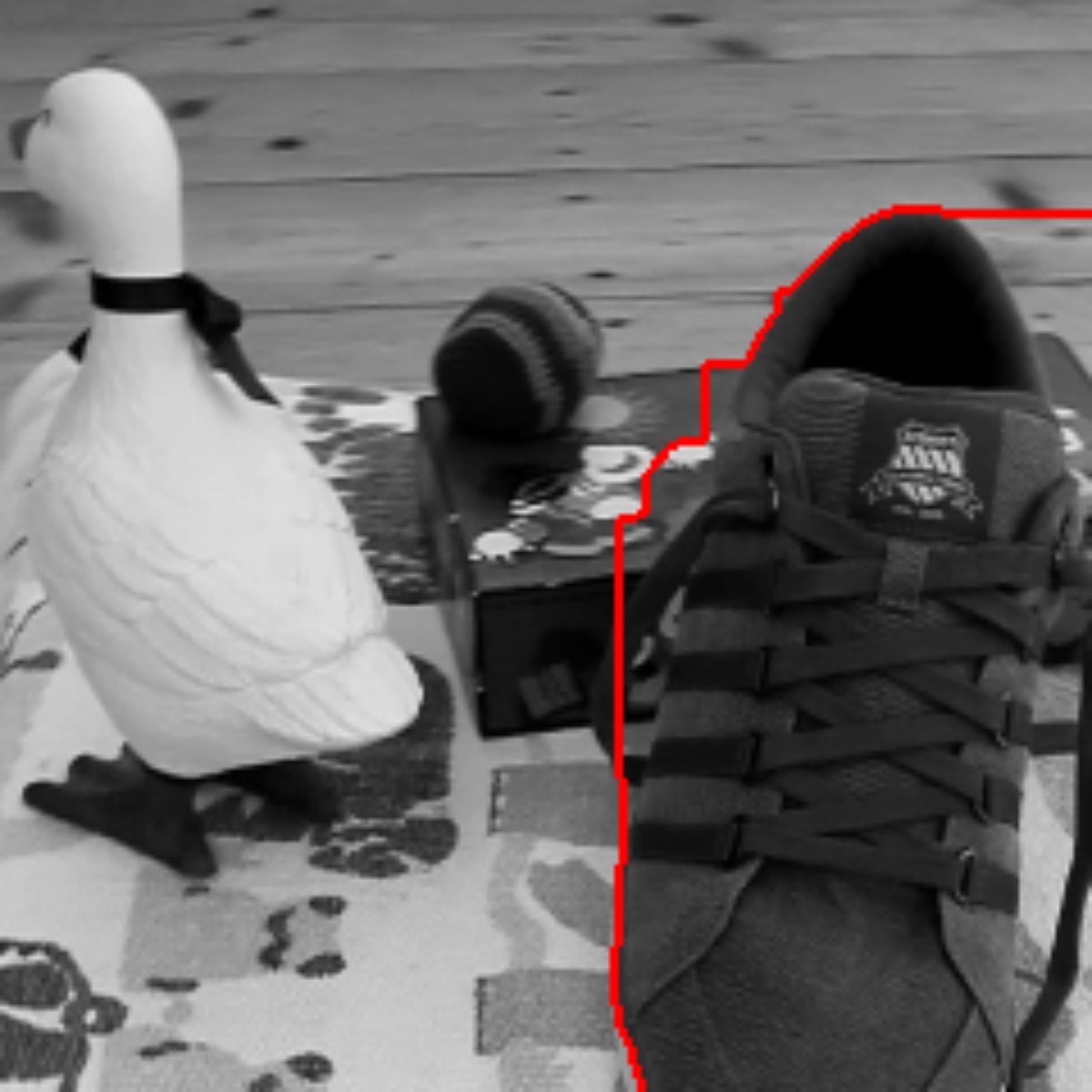}  
            \label{fig:shoecut1}
        \end{subfigure}
        \hspace{1cm}
        \begin{subfigure}[b]{0.2\linewidth}  
            \centering 
            \includegraphics[width=\linewidth]{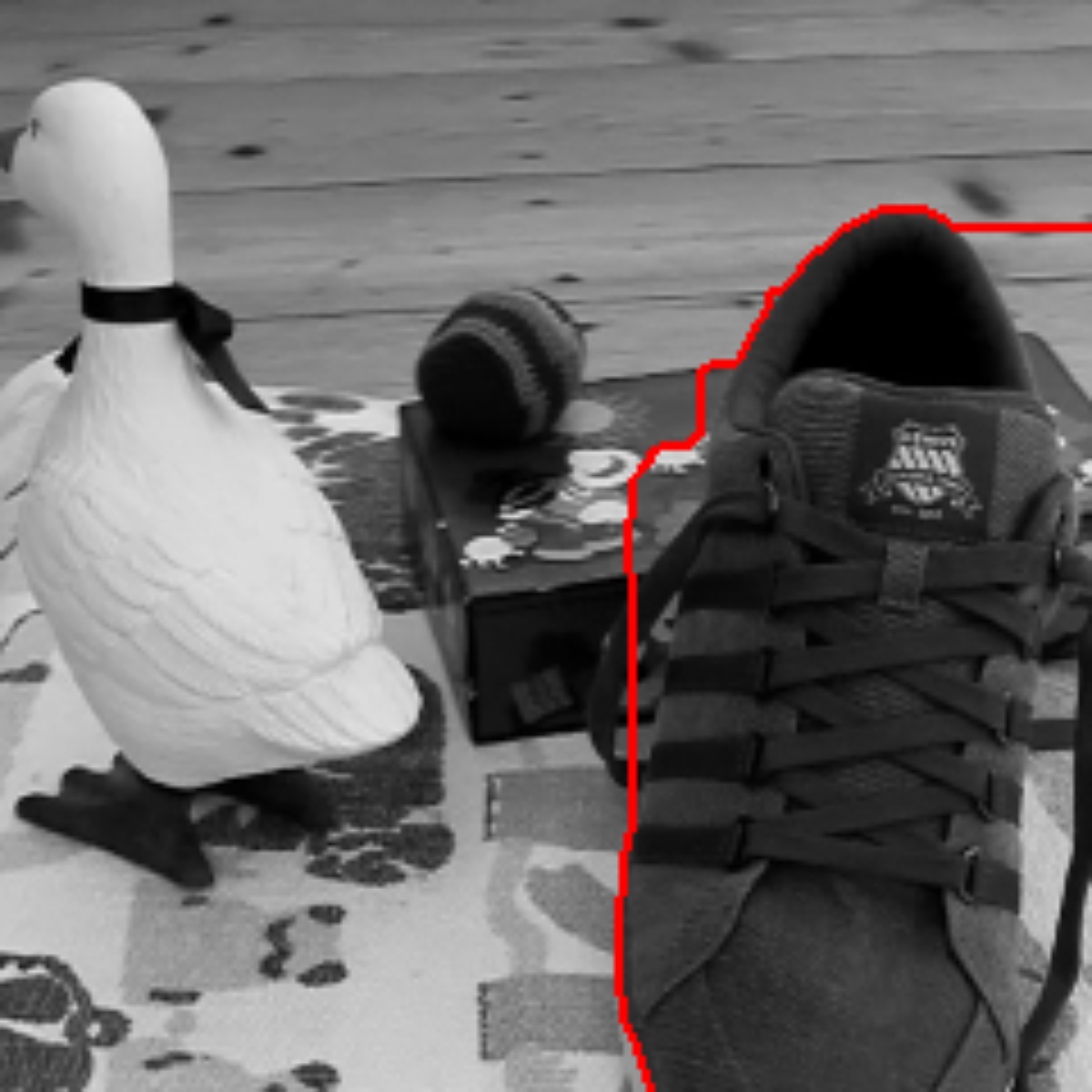}   
            \label{fig:shoecut2}
        \end{subfigure}
       \hspace{1cm}
            \begin{subfigure}[b]{0.2\linewidth}  
            \centering 
            \includegraphics[width=\linewidth]{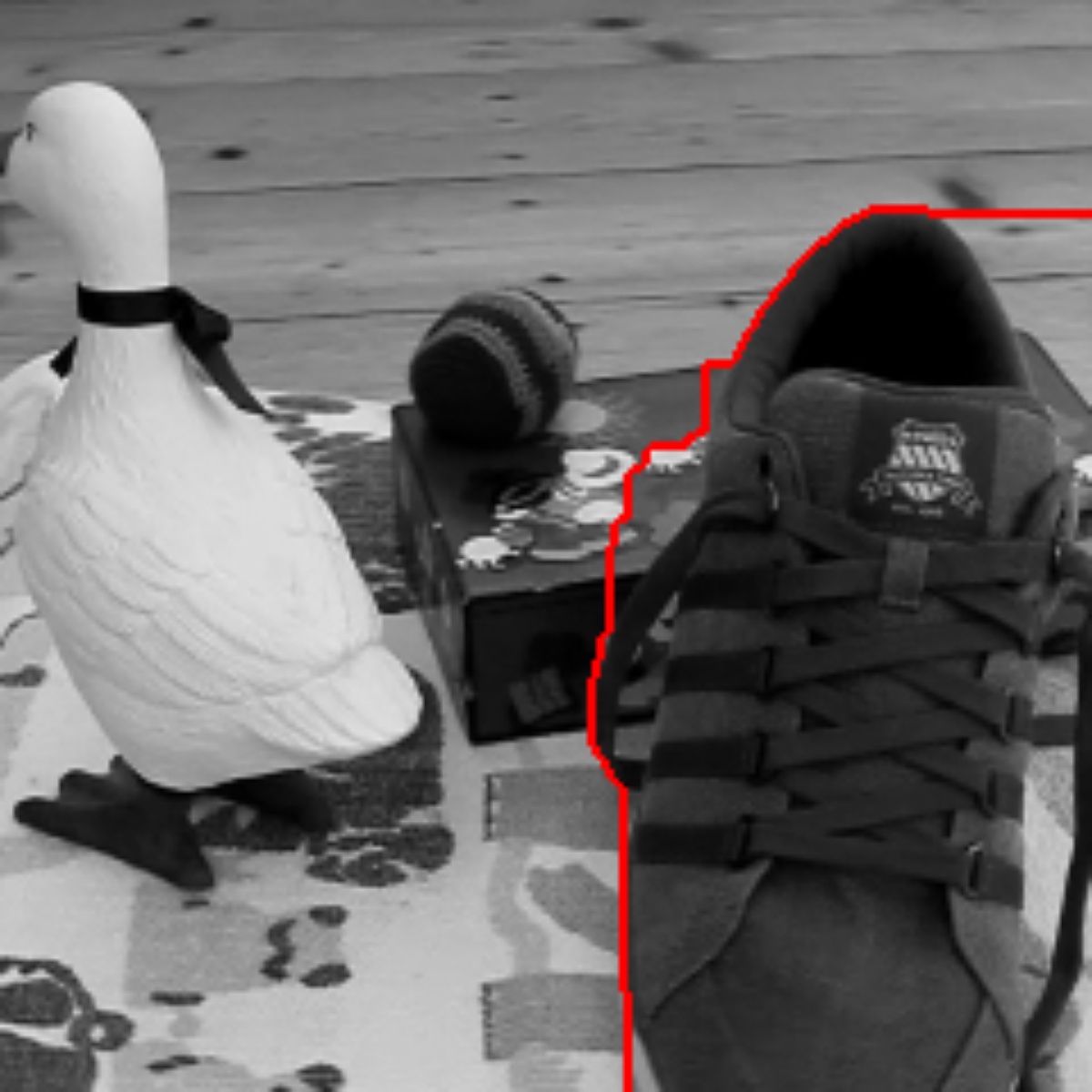}   
            \label{fig:shoecut3}
        \end{subfigure}
        \caption{\small  Cuts (red) on images chronologically evolving from the $240 \times 240$ pixel images from \textsc{Dog}.} 
        \label{fig:shoe_cuts}
% \vspace{-0.2in}
\end{figure}

\subsection{Full running time results}
\begin{table}[ht]
\centering
\caption{Average run-times (s) of cold-/warm-start Ford Fulkerson (FF) and Push-Relabel (PR)}
\label{table:running_time_ff_vs_pr_full}
\begin{tabular}{r|lllll}
\hline
    Image Group & FF cold-start & FF warm-start  &  PR cold-start & PR warm-start  \\ 
\hline
\textsc{Birdhouse} $30 \times 30$ & 0.80 & 0.51 & 0.05 & 0.06\\
\textsc{Head} $30 \times 30$ & 0.62 & 0.43 & 0.05 & 0.05 \\
\textsc{Shoe} $30 \times 30$ & 0.65 & 0.39 & 0.07 & 0.06 \\
\textsc{Dog} $30 \times 30$ & 0.69 & 0.32 & 0.10 & 0.11 \\
\hline
\hline
\textsc{Birdhouse} $60 \times 60$ & 8.22 & 3.25 & 0.30 & 0.45\\
\textsc{Head} $60 \times 60$ & 9.36 & 4.10 & 0.50 & 0.50 \\
\textsc{Shoe} $60 \times 60$ & 8.09 & 3.04 & 0.69 & 0.47 \\
\textsc{Dog} $60 \times 60$ & 21.91 & 6.73 & 0.76 & 0.95 \\
\hline
\hline
\textsc{Birdhouse} $120\times120$ & 109.06  &  37.31 &  5.42 & 4.98 \\
    \textsc{Head} $120\times120$& 101.79 & 28.43  & 5.90 & 5.92 \\
    \textsc{Shoe} $120\times120$& 98.95 & 30.44  & 6.44 & 3.74\\
    \textsc{Dog} $120\times120$ & 190.36 & 38.08   & 6.76 & 6.38 \\ 
\hline
\hline
\textsc{Birdhouse} $240 \times 240$ & NA & 400.19 & 60.67 & 55.68 \\
\textsc{Head} $240 \times 240$ & NA & 374.79 & 32.46 & 31.00\\
\textsc{Shoe} $240 \times 240$ & NA & 338.05 & 69.29 &  35.57 \\
\textsc{Dog} $240 \times 240$ & NA & 459.48 & 73.76 & 52.42\\
\hline
\hline
\textsc{Birdhouse} $480\times480$ & NA  &  NA &  604.54 & 502.58 \\
\textsc{Head} $480\times480$ & NA  &  NA &  365.25 & 285.75 \\
\textsc{Shoe} $480\times480$ & NA  &  NA &  756.77 & 364.42 \\
\textsc{Dog} $480\times480$ & NA  &  NA &  834.63 & 363.41 \\
\hline
\end{tabular}
\end{table}

 Table \ref{table:running_time_ff_vs_pr_full} compares the running time of cold-/warm-start Ford-Fulkerson as implemented in \cite{davies23b} against Push-Relabel on all data sizes and all image groups. The experiments were performed with the same computing configuration environment. One can see Push-Relabel greatly outperforms on the same image size, allowing us to collect run-time statistics on images of sizes up to $480 \times 480$ pixels, which we could not do with implementations of Ford-Fulkerson, due to its slow run-time.

\end{document}